\documentclass{mscs}
\usepackage{amssymb}
\usepackage{epsfig}

\newcommand{\ov}{\overline}
\newcommand{\sqeq}{\sqsubseteq}
\newcommand{\la}{\lambda}
\newcommand{\cc}{\circledcirc}
\newcommand{\adr}{{{\mathbb A}\,}}

\newcommand{\vpump}{\Uparrow}
\newcommand{\compdown}[2]{{#1}_{|#2}}
\newcommand{\bp}{{{\mathbb B}\hspace{.2ex}}}
\newcommand{\bpr}{\bp_{\varepsilon}}
\newcommand{\ebp}{\emptyset_{\,{\mathbb B}\,}}
\newcommand{\nat}{{{\mathbb N}\,}}
\newcommand{\brp}{{{\mathbb R}\,}}
\newcommand{\linear}{{{\mathbb F}\,}}
\newcommand{\seq}{{\mathbb S}}
\newcommand{\sign}{{\mathfrak S}}
\newcommand{\shadow}{{\mathfrak R}}
\newcommand{\lamt}{\Lambda_{\mbox{\rm\scriptsize NF}}}
\newcommand{\trite}{\vdash_{\mbox{\sf\scriptsize BB'IW}\,}}
\newcommand{\free}{\mbox{\sf Free}}
\newcommand{\sub}{\mbox{\sf Sub}}

\newcommand{\dom}{{\mbox{\rm dom}}}
\newcommand{\form}{\Omega}

\newtheorem{theorem}{Theorem}[section]
\newtheorem{lemma}[theorem]{Lemma}

\newtheorem{proposition}[theorem]{Proposition}

\newtheorem{definition}[theorem]{Definition}

\date{19 June 2010; Revised 6 March 2012}

\begin{document}
\title[Ticket Entailment is decidable]{Ticket Entailment is decidable}
\author[V. Padovani]
{V\ls I\ls N\ls C\ls E\ls N\ls T\ns P\ls A\ls D\ls O\ls V\ls A\ls N\ls I\\
Equipe Preuves, Programmes et Syst\`emes\\
 Universit\'e Paris VII - Denis Diderot\\
Case 7014\\
75205 PARIS Cedex 13\\
{\tt padovani@pps.jussieu.fr}}
\maketitle
\begin{abstract}
We prove the decidability of the logic $T_\to$ of Ticket Entailment. Raised by Anderson and Belnap within the framework of relevance logic, this question is equivalent to the question of the decidability of type inhabitation in simply-typed combinatory logic with the partial basis ${\sf BB'IW}$. We solve the equivalent problem of type inhabitation for the restriction of simply-typed lambda-calculus to hereditarily right-maximal terms.
\end{abstract}
The partial bases built upon the atomic combinators ${\sf B}$, ${\sf B'}$, ${\sf C}$, ${\sf I}$, ${\sf K}$, ${\sf W}$  of combinatory logic are well-known for being closely connected with propositional logic. The types of their combinators form the axioms of implicational logic systems that have been studied for well over 70 years \cite{TriggHB1994}.
The partial basis ${\sf BB'IW}$ corresponds, via the types of its combinators, to
 the system $T_\to$ of {\em Ticket Entailment} introduced and motivated in \cite{AndersonBelnap1975,AndersonBelnapDunn1990}. The system $T_\to$ consists of modus ponens and four axiom schemes that range over the following types for each atomic combinator:
\begin{itemize}
\item ${\sf B}$ :  $(\chi\to\psi)\to((\phi\to\chi)\to(\phi\to\psi))$
\item ${\sf B'}$ :  $(\phi\to\chi)\to((\chi\to\psi)\to(\phi\to\psi))$
\item ${\sf I}$ :  $\phi\to\phi$
\item ${\sf W}$ :  $(\phi\to(\phi\to\chi))\to(\phi\to\chi)$
\end{itemize}
The type inhabitation problem for ${\sf BB'IW}$ is the problem of deciding for a given type whether there exists within this basis a combinator of this type. This problem is equivalent to the problem of deciding whether a given formula can be derived in $T_\to$.

Surprisingly, the question of the decidability of $T_\to$ has remained unsolved since it was raised in \cite{AndersonBelnap1975}, although the
 problem has been thoroughly explored within the framework of relevance logic with
proofs of decidability and undecidability for several related systems. 
For instance the system $R_\to$ of {\em Relevant Implication} (which corresponds to the basis ${\sf B C I W}$) and the system $E_\to$ of {\em Entailment} \cite{AndersonBelnap1975} are both decidable \cite{Kripke1959} whereas the extensions $R$, $E$, $T$ of $R_\to$, $E_\to$, $T_\to$  to a larger set of connectives ($\to$, $\wedge$, $\vee$) are undecidable \cite{Urquhart1984}.

In 2004, a partial decidability result for the type inhabitation problem was proposed in \cite{BrodaDF2004} for a restricted class of formulas -- the class of {\em $1$-unary formulas} in which every maximal negative subformula is of arity at most 1. Broda, Dams, Finger and Silva e Silva's approach is based on a translation of the problem into a type inhabitation problem for the {\em hereditary right-maximal} (HRM) terms of lambda calculus \cite{TriggHB1994,Bunder1996,BrodaDF2004}. The closed HRM-terms form the closure under $\beta$-reduction of all translations of ${\sf BB'IW}$-terms, accordingly the type inhabitation problem within the basis ${\sf BB'IW}$ is equivalent to the type inhabitation problem for HRM-terms.

 We use in this paper the same approach as  Broda, Dams, Finger and Silva e Silva's. We prove that the type inhabitation problem for HRM-terms is decidable, and conclude that the logic $T_\to$ is decidable\footnote{In the course of the publication of this article, we heard of a work in progress by  Katalin Bimb\`o and Michael Dunn towards a solution that is seemingly based on a different approach.}. 
\subsection*{Summary}
In Section 1,
 we recall the definition of hereditarily right-maximal terms and the equivalence between the decidability of type inhabitation for ${\sf BB'IW}$ and the decidability of type inhabitation for HRM-terms. The principle of our proof is depicted on Figure~\ref{fig_proof_principle}. 

In Sections  2 and 3 we provide for each formula $\phi$ a partial characterisation of the inhabitants of $\phi$ in normal form and of minimal size. We show that all those inhabitants belong to two larger sets of terms, the set of {\em compact} and {\em locally compact} inhabitants~of~$\phi$. 
 
In Section 4 we show how to associate, with each locally compact inhabitant $M$ of a formula $\phi$, a labelled tree with the same tree structure as $M$. We call this tree the {\em shadow} of $M$.  We define for shadows the analogue of compactness for terms and prove that the shadow of a compact term is itself compact.

Finally, in Section \ref{sect_finite}, we prove that for each formula $\phi$ the set of all compact shadows of inhabitants of $\phi$ is a finite set (hence the set of compact inhabitants of $\phi$ is also a finite set), and that this set is effectively computable from $\phi$. The proof appeals to Higman Theorem and Kruskal Theorem -- more precisely, to Melli\`es' Axiomatic Kruskal Theorem. 

The decidability of the type inhabitation problem for HRM-terms and the decidability of $T_\to$ follow from this last key result: given an arbitrary formula $\phi$, this formula is inhabited if and only if there exists a compact shadow  with the same tree structure as an inhabitant of $\phi$, and our key lemma proves that the existence of such a shadow is decidable. 
\begin{figure}
\begin{center}
\epsfig{scale=1.2, file=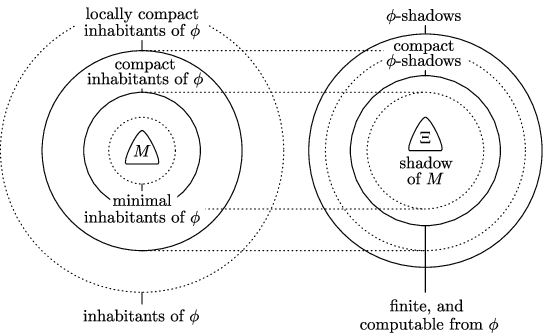}
\end{center}
\caption{\label{fig_proof_principle} Principle of the proof of decidability of type inhabitation for HRM-terms.}
\end{figure}
\section*{Preliminaries}
The first section of this paper assumes some familiarity with pure and simply-typed lambda-calculus and with the usual notions of $\alpha$-conversion, $\beta$-reduction and $\beta$-normal form \cite{Barendregt1984,Krivine1993}. The last three notions are not essential to our discussion, as we later focus exclusively on a particular set of simply-typed terms in $\beta$-normal form. We shall briefly recall the definitions and results used  in Section 1.

 The set of {\em terms of pure lambda-calculus} ($\lambda$-{\em terms}) is inductively defined by:
\begin{itemize}
\item every variable $x$ is a $\lambda$-term,
\item if $M$ is a $\lambda$-term and $x$ is a variable, then $(\lambda xM)$ is a $\lambda$-term,
\item if $M, N$ are $\lambda$-terms, then $(MN)$ is a $\lambda$-term.
\end{itemize}
Terms yielded by the second and third rules are called {\em abstractions} and {\em applications} respectively. The parentheses surrounding applications and abstractions are often omitted if unambiguous. 
We let $\la x_1\dots x_n.M N_1\dots N_p$ abbreviate $(\la x_1(\dots (\la x_n(((MN_1)\dots)N_p))\dots))$. 
For instance, $\la xy.x(xy) z$ stands for $(\la x(\la y((x(xy))z)))$.

 The {\em bound variables} of $M$ are all $x$ such that $\la x$ occurs in $M$. A variable $x$ is {\em free} in $M$ if and only:
\begin{itemize}
\item $M = x$, or,
\item $M =\la y.N$, $y\neq x$ and $x$ is free in $N$, or,
\item $M = NP$ and $x$ is free in $N$ or free in $P$.
\end{itemize}
A {\em closed} term is a term with no free variables. The {\em raw substitution of $N$ for $x$ in $M$}, written $M\langle x\leftarrow N\rangle$, is the term obtained by substituting $N$ for every free occurrence of $x$ in $M$ (every occurrence of $x$ that is not in the scope of a $\la x$). We require this substitution to avoid variable capture (for all $y$ free in $N$, no free occurrence of $x$ in $M$ is in the scope of a $\la y$):
\begin{itemize}
\item  if $y = x$, then $y\langle x\leftarrow N\rangle$ is equal to $N$, otherwise it is equal to $y$,
\item  $(\lambda x.M)\langle x\leftarrow N\rangle = \lambda x.M$,
\item if $y\neq x$ and $y$ is free in $N$, then  $(\lambda y.M)\langle x\leftarrow N\rangle$ is undefined,
\item if $y\neq x$, $y$ is not free in $N$ and $M\langle x\leftarrow N\rangle = M'$, then $(\la y.M)\langle x\leftarrow N\rangle = \la y.M'$,
\item if $M_1\langle x\leftarrow N\rangle = M'_1$ and  $M_2\langle x\leftarrow N\rangle = M'_2$, then  $(M_1M_2)\langle x\leftarrow N\rangle = (M'_1M'_2)$.
\end{itemize}
The {\em $\alpha$-conversion} is defined as the least binary relation $\equiv_\alpha$ such that:
\begin{itemize}
\item $x\equiv_\alpha x$,
\item if $M\equiv_\alpha M'$, $y$ is not free in $M'$ and $M'\langle x\leftarrow y\rangle = M''$, then $(\lambda x.M)\equiv_\alpha (\lambda y.M'')$
\item  if $M_1\equiv_\alpha M'_1$ and $M_2\equiv_\alpha M'_2$, then $(M_1M_2)\equiv_\alpha (M'_1M'_2)$. 
\end{itemize}
For instance $\la x.y\equiv_\alpha \la z.y\not\equiv_\alpha \lambda y.y$. It is a common practice to consider $\la$-terms up to $\alpha$-conversion, however we will not follow this practice in our exposition. 

The {\em  $\beta$-reduction} is the least binary relation $\beta$ satisfying:
\begin{itemize}
\item if $M\equiv_\alpha (\la x.N) P$ and $N\langle x\leftarrow P\rangle = N'$, then  $M\beta N'$.
\item if $M\beta M'$, then $(\la x.M)\beta(\la x.M')$, $(MN)\beta(M'N)$ and $(NM)\beta(NM')$.
\end{itemize} 
In the first rule, $x$ is not necessarily free in $N$, so we may have $N = N'$ -- in particular, free variables may disappear in the process of reduction.

We write $\beta^*$ for the reflexive and transitive closure of $\beta$. A term $M$ is {\em in $\beta$-normal form} -- or {\em $\beta$-normal} -- if there is no $M'$ such that $M\beta M'$. A term $M$ is {\em normalising} if there is a normal $N$ -- called {\em normal form} of $M$ -- such that $M\beta^*N$. It is {\em strongly normalising} if there is no infinite sequence $M = M_0\beta M_1\beta M_2\dots$ 

 It is well-known that $\beta$-conversion enjoys the {\em Church-Rosser property}: if $M\beta^* N$ and $M\beta^* N'$, then there exist two $\alpha$-convertible $P, P'$ such that $N\beta^* P$ and   $N'\beta^* P'$. As a consequence, if a term is normalising then its normal form is unique up to $\alpha$-conversion.

The judgment ``assuming $x_1,\dots, x_n$ are of types $\psi_1,\dots \psi_n$, the term $M$ is of type $\phi$'', written $\{x_1:\psi_1,\dots, x_n:\psi_n\}\vdash M:\phi$, where $\psi_1,\dots,\psi_n,\phi$ are formulas of propositional calculus and $x_1,\dots, x_n$ are distinct variables, is defined by:
\begin{itemize}
\item  $\Gamma\vdash x:\psi$ for each $x:\psi\in\Gamma$,
\item  if $\Gamma\cup\{x:\phi\}\vdash M:\psi$, then $\Gamma\vdash \la x.M:\phi\to\psi$.
\item  if $\Gamma\vdash M:\phi\to\psi$ and  $\Gamma\vdash N:\phi$, then $\Gamma\vdash (M N):\psi$
\end{itemize} 
The {\em simply-typable terms} are all $M$ for which there exist $\Gamma$, $\phi$ such that $\Gamma\vdash M:\phi$. Note that $\Gamma$ contains all variables free in $M$. 
The following properties are well-known:
\begin{enumerate} 
\item  (Strong normalisation)  If $\Gamma\vdash M:\phi$, then $M$ is strongly normalising.
\item (Subject reduction) If $\Gamma\vdash M:\phi$ and $M\beta N$, then $\Gamma\vdash N:\phi$.
\end{enumerate}
\section{From ${\sf BB'IW}$ to simply-typed lambda-calculus}\label{sect_lambda}
The aim of this first section is to provide a precise characterisation of simply-typable terms that are typable with inhabited types in ${\sf BB'IW}$, so as to transform the problem of type inhabitation in ${\sf BB'IW}$  into a type inhabitation problem in lambda-calculus. 
The types of atomic combinators in ${\sf B B' I W}$ are also types for their respective counterparts $\la fgx.f(gx)$,  $\la fgx.g(fx)$, $\la x.x$, $\la hx.h x x$ in lambda-calculus, hence to each inhabited type $\phi$ in ${\sf B B' I W}$  corresponds at least one closed $\la$-term of type $\phi$. Moreover, subject reduction and strong normalisation (see above) also ensure the existence of a closed normal $\la$-term of type $\phi$. What we lack is a criterion to distinguish amongst all typed normal forms the ones that are reducts of translations of combinators within  ${\sf B B' I W}$.

The material and the results of this section are not new \cite{Bunder1996,BrodaDF2004}. The reader may as well skip the contents of Sections \ref{sect_preserve} and \ref{sect_equiv} entirely, accept Lemma \ref{main_equiv} then go on with the study of stable parts and blueprints in Section \ref{sect_blue}.

The definition of hereditarily right-maximal terms is an adaptation of the definition given in \cite{Bunder1996}. The proof of Lemma \ref{subject} (subject reduction for HRM-terms) is similar to the proof of Property 2.4, p.375 in \cite{BrodaDF2004}. 
The right-to-left implication of Lemma~\ref{main_equiv} can be deduced from Property 2.20, p.390 in \cite{BrodaDF2004}, although our proof method seems to be simpler.

\subsection{Lambda-calculus}\label{la_calc}
Let ${\cal X}$ be a countably infinite set of variables $x,y,z\dots$ together with a one-to-one function ${\cal O}$ from ${\cal X}$ to $\nat$. For all $x,y$ in ${\cal X}$, we write $x < y$ when ${\cal O}(x) < {\cal O}(y)$.
Throughout the sequel, by {\em term} we always mean a term of lambda-calculus built over those variables. For each term $M$, we write $\free(M)$ for the strictly increasing sequence of all free variables~of~$M$. 

Terms are {\em not} identified modulo $\alpha$-conversion - apart from Section 1, all considered terms will be in normal form, and the Greek letters $\alpha$, $\beta$ will be even used with new meaning at the beginning of Section~2. 
We adopt however the usual convention according to which two distinct $\lambda$'s may not bound the same variable in a term, and no variable can be simultaneously free and bound in the same term. 
\subsection{Hereditarily right-maximal terms}
\begin{definition} The set of {\em hereditarily right-maximal} (HRM) terms is inductively defined as follows:
\begin{enumerate}
\item Each variable $x$ is HRM.
\item If $M$ is HRM and $x$ is the greatest free variable of $M$ 
then $\la x.M$ is HRM.
\item If $M,  N$ are HRM, and for each free variable $x$ of $M$ there exists a free variable $y$ of $N$ such that $x\leq y$, then $(MN)$ is HRM.
\end{enumerate}
\end{definition}
The second rule ensures that all HRM-terms are $\lambda_I$-terms, that is, terms in which every subterm $\la x.M$ is such that $x$ is free in $M$. As a consequence the set of free variables of an HRM-term is preserved under $\beta$-reduction. As we shall see below (Lemma \ref{subject}), right-maximality can also be preserved at the cost of appropriate bound variable~renamings.

In the third rule, if $N$ is closed then so is $M$. When $M$ and $N$ are non-closed terms, the greatest free variable of $M$ is less than or equal to  the greatest free variable of $N$.
For instance, if $f < g < x$ and $h < x$, then $\lambda fgx.f(gx)$,  $\lambda fgx.g(fx)$, $\la x.x$, $\la hx.hxx$ are HRM, whereas $\la yz.z y$ is not, no matter if $y < z$ or $y > z$.
\begin{definition}
Let $\Omega$ be a function mapping each variable to a formula, in such a way that $\Omega^{-1}(\phi)$ is an infinite set for each $\phi$. We extend this function to the set of all strictly increasing finite sequences of variables, letting
 $\form(x_1,\dots,x_n) =(\form(x_1),\dots,\form(x_n))$. 
\end{definition}
\begin{definition}The judgment $M:\phi$, in words ``$M$ is of type $\phi$ w.r.t $\Omega$'', is defined by:
\begin{itemize}
\item if $\Omega(x) = \phi$, then  $x:\phi$,
\item if $x:\chi$, $M :\psi$ and $\la x.M$ is HRM, then  $\la x.M :\chi\to\psi$,
\item if $M :\chi\to\psi$, $N:\chi$ and $(M N)$ is HRM, then $(MN):\psi$.
\end{itemize}
\end{definition}
The function $\Omega$ will remain fixed throughout our exposition. Accordingly the type of a term $M$ w.r.t $\Omega$ will be called {\em the} type of $M$, without any further reference to the choice of $\Omega$. Note that every typed term is HRM. 
\begin{definition}
We write $\lamt$ for the set of all typed terms in $\beta$-normal form. We call {\em $\lamt$-inhabitant of $\phi$} every closed term {$M\in\lamt$} of type $\phi$.
\end{definition}
 The next lemma is the well-known subformula property of simply-typed lambda-calculus:
\begin{lemma} \label{normal_propto} (Subformula Property) Let $M$ be a $\lamt$-inhabitant of $\phi$. The types of the subterms of $M$ are subformulas of $\phi$.
\end{lemma}
\subsection{Subject reduction of hereditarily right-maximal terms}\label{sect_preserve}
\begin{lemma}\label{subject} Suppose there exists a closed $M:\phi$. Then $\phi$ is $\lamt$-inhabited.
\end{lemma}
\begin{proof} (1) We leave to the reader the proof of the fact that for every variable $y$ and for every $N:\phi$, there exists $N'\equiv_\alpha N$ such that $N':\phi$ and every bound variable of $N'$ is strictly greater than $y$. 

(2) We prove the following proposition by induction on $P$.
 Let $P, Q$ be typed HRM-terms. Suppose:
\begin{itemize}
\item $x$ and $Q$ are of the same type, 
\item if $Q$ is closed and $x\in\free(P)$, then $x = \min(\free(P))$
\item if $Q$ is not closed, then for all $z\in\free(P)$:

\ \ \ \ \  if $z < x$ then $z \leq \max(\free(Q))$,

\ \ \ \ \  if $x < z$ then $\max(\free(Q)) < z$.
\item if $Q$ is not closed, then $\max(\free(Q)) < z$ for all bound variables $z$ of $P$.
\end{itemize}
Then $R = P\langle x\leftarrow Q\rangle$ is defined, HRM and of the same type as $P$. The proposition is clear if $P$ is a variable.

Suppose $P = \lambda z.P'$. Then $\free(P') = \free(P)\cdot(z)$. By induction hypothesis  $R' = P'\langle x\leftarrow Q\rangle$ is defined, HRM and of the same type as $P'$. The variable $z$ is still the greatest free variable of $R'$ and $z$ is not free in $Q$, hence $R = \la z.R'$.

Suppose  $P = (P_1P_2)$. By induction hypothesis $R_i = P_i\langle x\leftarrow Q\rangle$ is defined, HRM and of the same type as $P_i$ for each $i\in\{1,2\}$. It remains to check that $R = (R_1R_2)$ is HRM. Assume $x$ is free in $P$ and $P_1$ is not closed. 

Suppose $\max(\free(P_1)) > x$. Then $\max(\free(P_1)) = \max(\free(R_1))\leq \max(\free(P_2)) = \max(\free(R_2))$. 

Suppose $\max(\free(P_1)) < x$. The term $Q$ cannot be closed, and $\max(\free(P_1)) = \max(\free(R_1))\leq \max(\free(Q))$. We have either $\max(\free(P_2)) = x$ and $\max(\free(R_2)) = \max(\free(Q))$, or $\max(\free(P_2)) > x$ and $\max(\free(P_2)) = \max(\free(R_2))$.

Otherwise $\max(\free(P_1)) = x$. Suppose $\max(\free(P_2)) > x$.  Then $\max(\free(P_2)) = \max(\free(R_2))$. If $Q$ is closed, then $\free(P_1) = (x)$ and $R_1$ is closed. Otherwise we have $\max(\free(R_1)) = \max(\free(Q)) \leq \max(\free(P_2))$. The remaining case is $\max(\free(P_2)) = x$. If $Q$ is closed then $\free(P_1) = \free(P_2) = (x)$ and $R_1$, $R_2$ are closed. Otherwise 
 $\max(\free(R_1)) = \max(\free(R_2) = \max(\free(Q))$.

(3) Assume $N:\phi$ and $N$ is not in normal form. We prove by induction on $N$ the existence of $N':\phi$ such that $N\beta N'$. If $N = \la x.P$, or if $N = (N_1 N_2)$ with $N_1$ or $N_2$ not in normal form, then the existence of $N'$ follows from the induction hypothesis and the fact that $\beta$-reduction preserves the set of free variables of an HRM-term. Otherwise $N = (\la x.P) Q$ where for each free variable $z$ of $\la x.P$, we have $z < x$ and there exists a free variable $y$ of $Q$ such that $z < y$. By (1) there exists $P'\equiv_\alpha P$ such that $P':\phi$ and no bound variable of $P'$ is less than or equal to a free variable of $Q$. The variable $x$ is the greatest free variable of $P'$. By (2), the term $N' = P'\langle x\leftarrow Q\rangle$ is well-defined, HRM and of the type $\phi$. Moreover $N\beta N'$.

 (4) We now prove the lemma. The term $M$ is a simply-typable HRM-term. The strong normalisation property implies the existence of a normal form $N$ of $M$. The term $N$ is still a closed term. By (1), there exists
$N'\equiv_\alpha N$ such that $N' :\phi$, that is, $\phi$ is $\lamt$-inhabited,
\end{proof}
\subsection{Equivalence between inhabitation in ${\sf BB'IW}$ and $\lamt$-inhabitation}\label{sect_equiv}
\noindent In the next three lemmas by $\phi_1\dots\phi_n\to\psi$ we mean the formula $(\phi_1\to(\dots (\phi_n\to\psi)\dots))$ if $n >0$, and otherwise the formula $\psi$. We write $\trite\phi$ for the judgment ``there exists within the basis ${\sf BB'IW}$ a combinator of type $\phi$''.
\begin{lemma}\label{step_1} If $\trite\phi$, then $\phi$ is $\lamt$-inhabited.
\end{lemma}
\begin{proof} If $f < g < x$ and $h < x$, then $\la x.x$, $\la f g x.f (gx)$,  $\la f g x.g (fx)$ and $\la hx.hxx$ are HRM. For each type $\phi$ of an atomic combinator, the variables $f, g, h, x$ can be chosen so that one of those terms is of type $\phi$. The set of all formulas $\phi$ for which there exists a closed $M$ of type $\phi$ is closed under modus ponens.  By Lemma \ref{subject}, every such formula is $\lamt$-inhabited.
\end{proof}
\begin{lemma}\label{extend} If $\trite \chi\to\psi$, then 
$\trite (\phi_1\dots\phi_n\to\chi)\to (\phi_1\dots\phi_n\to\psi)$ for all~$\phi_1,\dots,\phi_n$.
\end{lemma}
\begin{proof} By induction on $n$, using left-applications of ${\sf B}$.
\end{proof}
\begin{lemma}\label{apply_propto} Suppose $(i_1,\dots, i_n)$, $(j_1, \dots, j_m)$,  $(k_1,\dots, k_p)$ are strictly increasing sequences of integers, $\{k_1,\dots, k_p\} = \{i_1,\dots, i_n,j_1,\dots, j_m\}$, $n = 0$ or ($n > 0$, $m > 0$, $i_n\leq j_m$). If
\begin{enumerate}
\item $\trite \omega_{i_1}\dots \omega_{i_n}\to (\chi\to\psi)$,
\item $\trite \omega_{j_1}\dots \omega_{j_m}\to \chi$,
\end{enumerate}
then  $\trite \omega_{k_1}\dots \omega_{k_p}\to \psi$.
\end{lemma}
\begin{proof} By induction on $n+ m$. The proposition is true if $n = m = 0$. Assume $n + m > 0$. Then $m > 0$.

Suppose $n = 0$. Then $(j_i, \dots, j_m) = (k_1,\dots, k_p)$. We have:\\
$\begin{array}{ll}
(\mbox{\rm i}) & \trite (\chi\to\psi)\to((\omega_{j_m}\to\chi)\to(\omega_{j_m}\to\psi))\\
(\mbox{\rm ii}) & \trite (\omega_{j_m}\to\chi)\to(\omega_{j_m}\to\psi)
\end{array}$\\
where: (i) is a type for ${\sf B}$; (ii) follows from (i), (1) and modus ponens. If $m = 1$ then $\trite \omega_{j_1}\to \psi$ follows from (ii), (2) and modus ponens. Otherwise $\trite \omega_{j_1}\dots \omega_{j_m}\to\psi$ follows from (ii), (2) and the induction hypothesis.

We now assume $n > 0$.
Suppose $m > 1$ and $i_n \leq j_{m-1}$. Then\\
$\begin{array}{ll}
(\mbox{\rm iii}) & \trite (\chi\to\psi)\to((\omega_{j_m}\to\chi)\to(\omega_{j_m}\to\psi))\\
(\mbox{\rm iv}) &\trite (\omega_{i_1}\dots \omega_{i_n}\to(\chi\to\psi))\to(\omega_{i_1}\dots \omega_{i_n}\to((\omega_{j_m}\to\chi)\to(\omega_{j_m}\to\psi)))\\
(\mbox{\rm v}) &\trite \omega_{i_1}\dots \omega_{i_n}\to((\omega_{j_m}\to\chi)\to(\omega_{j_m}\to\psi))
\end{array}$\\
where: (iii) is a type for ${\sf B}$; (iv) follows from (iii) and Lemma  \ref{extend}; (v) follows from (iv), (1) and modus ponens. We have
 $k_p = j_m$ and $\{k_1,\dots, k_{p-1}\} = \{i_1,\dots, i_{n},j_1,\dots, j_{m-1}\}$.  Since $i_n \leq j_{m-1}$, we have $\trite \omega_{k_1}\dots \omega_{k_{p-1}}\to(\omega_{j_m}\to\psi)$ by (v), (2) and the induction hypothesis.

Suppose $m = 1$ or ($m > 1$ and $i_n > j_{m-1}$). Then\\
$\begin{array}{ll}
(\mbox{\rm vi}) &\trite (\omega_{j_m}\to\chi)\to((\chi\to\psi)\to(\omega_{j_m}\to\psi))\\
(\mbox{\rm vii}) & \trite (\omega_{j_1}\dots\omega_{j_m}\to\chi)\to(\omega_{j_1}\dots\omega_{j_{m-1}}\to((\chi\to\psi)\to(\omega_{j_m}\to\psi)))\\
(\mbox{\rm viii}) & \trite \omega_{j_1}\dots\omega_{j_{m-1}}\to((\chi\to\psi)\to(\omega_{j_m}\to\psi))\\
(\mbox{\rm ix}) & \trite \omega_{n_1} \dots \omega_{n_{q}}\to(\omega_{j_m}\to\psi)
\end{array}$\\
where: (vi) is a type for ${\sf B'}$; (vii) follows from (vi) and  Lemma \ref{extend}; (viii) follows from (vii), (2) and modus ponens;  $\{n_1,\dots, n_{q}\} = \{j_1,\dots, j_{m-1},i_1,\dots, i_{n}\}$; (ix) follows from (viii), (1) and the induction hypothesis. 
If $j_m > i_n$, then  $(n_1,\dots, n_q, j_m) = (k_1,\dots, k_p)$. Otherwise $j_m = i_n$, $n_{q} = i_n$, $(n_1,\dots n_{q}) = (k_1,\dots, k_p)$ and\\
$\begin{array}{ll}
(\mbox{\rm x}) & \trite \omega_{k_1} \dots \omega_{k_{p-1}}\to(\omega_{i_n}\to(\omega_{i_n}\to\psi))\\
(\mbox{\rm xi}) &\trite (\omega_{i_n}\to(\omega_{i_n}\to\psi))\to(\omega_{i_n}\to\psi)\\
(\mbox{\rm xii}) &\trite (\omega_{k_1} \dots \omega_{k_{p-1}}\to(\omega_{i_n}\to(\omega_{i_n}\to\psi)))\to (\omega_{k_1} \dots \omega_{k_{p-1}}\to(\omega_{i_n}\to\psi))\\
(\mbox{\rm xiii}) & \trite \omega_{k_1} \dots \omega_{k_{p-1}}\to(\omega_{i_n}\to\psi)
\end{array}$\\
where: (x) is (ix); (xi) is a type for ${\sf W}$; (xii) follows from (xi) and Lemma \ref{extend}; (xiii) follows from (x), (xii) and modus ponens; (xiii) is  $\trite \omega_{k_1} \dots \omega_{k_{p}}\to\psi$.
\end{proof}
\begin{lemma}\label{main_equiv} For every formula $\phi$, we have $\trite\phi$ if and only if $\phi$ is $\lamt$-inhabited.
\end{lemma}
\begin{proof} The left to right implication is Lemma \ref{step_1}. Using Lemma~\ref{apply_propto} when $M$ is an application, an immediate induction on $M$ shows that if $M:\psi$, $\free(M) = (x_1,\dots, x_n)$ and $x_1:\chi_1,\dots, x_n:\chi_n$, then $\trite\chi_1\dots\chi_n\to\psi$
\end{proof}
\section{Stable parts and blueprints}\label{sect_blue}
The last lemma showed that the decidability of type inhabitation for ${\sf BB'IW}$  is equivalent to the decidability of $\lamt$-inhabitation. The sequel is devoted to the elaboration of a decision algorithm for the latter problem.

The problem we shall examine throughout Sections 2 and 3 is the following: if an inhabitant is not of minimal size, is there any way to transform it (with the help of grafts and/or another compression of some sort) into a smaller inhabitant of the same type? This question is not easy because we are dealing with a lambda-calculus restricted with strong structural constraints (righ-maximality). There are however simple situations in which an inhabitant is obviously not of minimal size.

Consider a $\lamt$-inhabitant $M$ and two subterms $N, P$ of $M$ such that $P$ is a strict subterm of $N$. Suppose:
\begin{itemize}
\item $N,P$ are applications of the same type or abstractions of the same type.
\item $\free(N) = X = (x_1,\dots, x_n)$, 
\item  $\free(P) = Y = (y^1_0,\dots,y^1_{p_1},\dots, y^n_0,\dots,y^n_{p_n})$
\item $\Omega(X) = (\chi_1,\dots, \chi_n)$,
\item  $\Omega(Y) = (\chi^1_0,\dots,\chi^1_{p_1},\dots, \chi^n_0,\dots,\chi^n_{p_n})$,
\item $\chi^i_j = \chi_i$ for each $i,j$.
\end{itemize} 
Then $M$ is not of minimal size. Indeed we can rename the free variables of $P$ (letting $\rho(y^i_j) = x_i$) so as to obtain a term $P'$ of the same size as $P$, of the same type and the same free variables as $N$. The subterm $N$ of $M$ can be replaced with $P'$ in $M$. The resulting term is a $\lamt$-inhabitant of the same type but of strictly smaller size.

This simple property is far from being enough to characterise the minimal inhabitants of a formula: there are indeed formulas with inhabitants of abitrary size in which this situation never occurs. What we need is a more flexible way to reduce the size of non-minimal inhabitants. In particular, we need a better understanding of our available freedom of action if we are to rename the free variables of a term -- possibly occurrence by occurrence -- and if we want to ensure that right-maximality is preserved. 
This section is devoted to the proof of two key lemmas that delimit this freedom.
\begin{itemize}
\item In Sections \ref{sect_tree}, \ref{subsect_bp} and \ref{subsect_bp} we show how to build from any term $M\in\lamt$ a partial tree labelled with formulas. This partial tree is called the {\em blueprint} of $M$. This blueprint can be seen as a synthesized version of $M$ that contains  all and only the information required to determine whether a (non-uniform) renaming of the free variables of $M$ will preserve hereditarily right-maximality.  
\item In Sections \ref{subsect_red_def} and \ref{subsect_extract} we introduce a rewriting relation on blueprints that allows one to ``extract'' sequences of formulas from a blueprint.
\item In section \ref{subsect_abstr_red} we prove our two key lemmas. 
Lemma \ref{abstract_extract} clarifies the link between the blueprints of $M$ and $\la x.M$ (provided both are in $\lamt$). This lemma proves in particular that the sequence of the types of the free variables of $M$ (that is, $\Omega(\free(M))$) can always be extracted from its blueprint. 
Lemma \ref{switch_var} shows that for {\em every} sequence of formulas $\overline\phi$ that can be extracted from the blueprint of $M$, there exists a (non-uniform) renaming of the free variables of $M$ that will produce a term $N$ of the same type and with the same blueprint as $M$, and such that $\Omega(\free(N)) = \ov \phi$.
\end{itemize}
As a continuation of our first example, let us examine the consequences of this last result. Consider again a $\lamt$-inhabitant $M$ and two subterms $N, P$ of $M$ such that $P$ is a strict subterm of $N$ and $N,P$ are applications of the same type or abstractions of the same type. Suppose:
\begin{itemize}
\item the sequence $\Omega(\free(N))$ can be extracted from the blueprint of  $P$.
\end{itemize} 
This situation is a generalization of the preceding one (in our first example $\Omega(X)$ could also be extracted from the blueprint of  $P$, see Definition \ref{def_linear}). The term $M$ is still not of minimal size. Indeed, we may use the second key lemma to prove the existence of (non-uniform) renaming of the free variables of $P$ that will produce a term $P'$ of the same type as $P$ such that $\free(P') = \free(N)$. The term $N$ can be replaced with $P'$ in $M$. 
\subsection{Partial trees and trees}\label{sect_tree}
\begin{definition}
Let $(\adr, \leq)$ be the set of all finite sequences over the set $\nat_+$ of natural numbers, ordered by prefix ordering. Elements of $\adr$ are called {\em addresses}. We call {\em partial tree} every function $\pi$ whose domain is a set of addresses. For each partial tree $\pi$ and for each address $a$, we let $\compdown{\pi}{a}$ denote the partial tree $c \mapsto \pi(a\cdot c)$ of domain $\{c\,|\,a\cdot c\in\dom(\pi)\}$. 
\end{definition}
\begin{definition}
For all partial trees $\pi, \pi'$ and for every address $a$, we let $\pi[a\leftarrow\pi']$ denote the partial tree $\pi''$ such that $\compdown{\pi''}{a} = \pi'$ and $\pi''(b) = \pi(b)$ for all $b\in\dom(\pi)$ such that $a\not\leq b$.
\end{definition}
\begin{definition}
A {\em tree domain} is a set $A\subseteq\adr$ such that for all $a\in A$: every prefix of $a$ is in $A$; for every integer $i > 0$, if $a\cdot(i)\in A$, then $a\cdot(j)\in A$ for each $j\in\{1,\dots, i -1\}$.  
A tree domain $A$ is {\em finitely branching} if and only if for each $a\in A$, there exists an $i > 0$
such that $a\cdot(i)$ is undefined.
We call {\em tree} every function whose domain is a tree domain.
\end{definition}
In the sequel terms will be freely identified with trees. We identify: $x$ with the tree mapping $\varepsilon$ to $x$; $\la x.M$ with the tree $\tau$ mapping $\varepsilon$ to $\la x$ and such that $\compdown{\tau}{(1)}$ is the tree of $M$; $(M_1M_2)$ with the tree $\tau$ mapping $\varepsilon$ to $@$ and such that  $\compdown{\tau}{(i)}$ is the tree of $M_i$ for each $i\in\{1,2\}$.
\subsection{Blueprints}\label{subsect_bp}
\begin{definition}\label{def_bp}
Let $\sign$ be the signature consisting of all formulas and all symbols of the form $@_\phi$ where $\phi$ is a formula. Each formula is considered as a symbol of null arity. Each  $@_\phi$ is of arity 2. 

We call {\em blueprint} every finite partial tree $\alpha: A\to\sign$ satisfying the following condition:  for each $a\in A$, if $\alpha(a) = @_\phi$, then $\compdown{\alpha}{a\cdot(1)}$ and  $\compdown{\alpha}{a\cdot(2)}$ are of non-empty domains. A {\em rooted} blueprint is a blueprint $\alpha$ such that $\varepsilon\in\dom(\alpha)$.

For each ${\cal S}\subseteq\sign$, we call {\em ${\cal S}$-blueprint} every blueprint whose image is a subset of ${\cal S}$. We write $\bp({\cal S})$ for the set of all ${\cal S}$-blueprints, and $\bpr({\cal S})$ for the set of all rooted ${\cal S}$-blueprints.
\end{definition}
\begin{definition}
For every blueprint $\alpha$ and every address $a$, the {\em relative depth of $a$ in $\alpha$} is the number of $b\in\dom(\alpha)$ such that $b < a$. The {\em relative depth of $\alpha$} is defined as $0$ if $\alpha$ is of empty domain, the maximal relative depth of an address in $\alpha$ otherwise.
\end{definition}
In the sequel the following notations will be used to denote blueprints (see Figure \ref{blueprint_cons}):
\begin{itemize}
\item $\ebp$ denotes the blueprint of empty domain.
\item we abbreviate $\varepsilon\mapsto\phi$ as $\phi$.
\item $@_\phi(\alpha_1,\alpha_2)$ denotes the (rooted) blueprint $\alpha$ such that $\alpha(\varepsilon) = \phi$, $\compdown{\alpha}{(1)} = \alpha_1$,  $\compdown{\alpha}{(2)} =~\alpha_1$.
\item for every sequence $\ov a = (a_1,\dots,a_k)$ of pairwise incomparable addresses, $\ast_{\ov a}(\alpha_1,\dots,\alpha_k)$  denotes the blueprint $\alpha$ of minimal domain such that $\compdown{\alpha}{a_i} = \alpha_i$ for each $i\in[1,\dots,k]$.
\item we let $\ast(\alpha_1,\dots,\alpha_k)$ denote the blueprint $\ast_{\ov a}(\alpha_1,\dots,\alpha_k)$ such that $\ov a = ((1),\dots,(k))$.
\end{itemize}
\begin{figure}
\begin{center}
\epsfig{scale=1.2, file=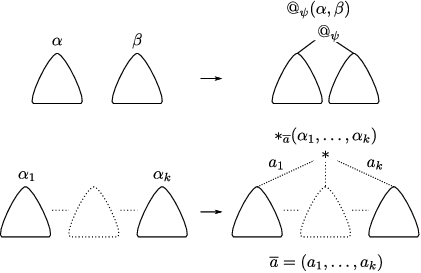}
\end{center}
\caption{\label{blueprint_cons} Construction of blueprints, with the notations of Section \ref{subsect_bp}. In the upper diagram, the blueprints $\alpha$ and $\beta$ must be non-empty. Although $\alpha_1,\dots,\alpha_k$ are displayed from left to right, the sequence $(a_1,\dots,a_k)$ needs not to be lexicographically ordered.}
\end{figure}
For each blueprint $\alpha$, the choice of $\ov a,\alpha_1,\dots,\alpha_k$ such that $\alpha = \ast_{\ov a}(\alpha_1,\dots,\alpha_k)$ is obviously not unique. The sequence $(\alpha_1,\dots,\alpha_k)$ may contain an arbitrary number of empty blueprints, hence the sequence $\ov a$ may be of arbitrary length. Also, $\alpha$ can be roooted (if $k = 1$, $a_1=\varepsilon$ and $\alpha_1$ is rooted) or empty (if $k = 0$ or $\alpha_1 = \dots = \alpha_k = \ebp$).  Those ambiguities will not be difficult to deal with, but they will require a few precautions in our proofs and definitions by induction on blueprints.
\subsection{Blueprint of a term}\label{subsect_bpt}
\begin{definition}\label{def_stable} For all $M\in\lamt$, 
the {\em stable part} of $M$ is the set of all  $a\in\dom(M)$ such that  $\free(\compdown{M}{a})\subseteq \free(M)$ and $\compdown{M}{a}$ is a variable or an application.
\end{definition}
It is easy to check that our conventions (no variable is simultaneously free and bound in a term) ensure that the stable part of a term does not depend on the choice of variable names. Since $M$ is in normal form, $M$ is of empty stable part if and only if it is closed. 
\begin{definition} For all $M\in\lamt$,  we call {\em blueprint of $M$} the function $\alpha$ mapping each $a$ in the stable part of $M$ to:
\begin{itemize}
\item $\psi$\ \ \  if $\compdown{M}{a}$ is a variable of type $\psi$,
\item $@_\psi$ if $\compdown{M}{a}$ is an application of type $\psi$.
\end{itemize}
We let $M\Vdash\alpha$ denote the judgment ``$M$ is of blueprint $\alpha$'' (Figure \ref{blueprint_stable}).
\end{definition}
\begin{figure}
\begin{center}
\epsfig{scale=1.2, file=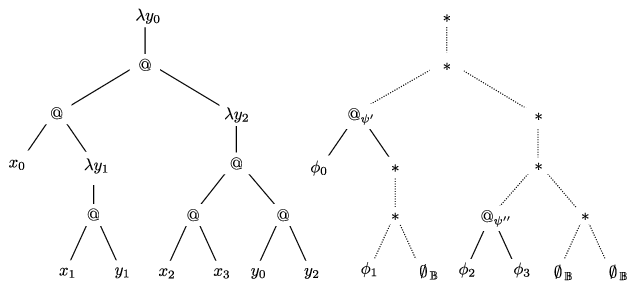}
\end{center}
\caption{\label{blueprint_stable} An element of $\lamt$ with its blueprint ($x_0 < x_1 < y_1$, $x_2 < x_3 < y_0 < y_2$, $x_1 < y_0 < y_2$).}
\end{figure}

If $M = (M_1 M_2)\in\lamt$, $M:\phi$, $M_1\Vdash\alpha_1$, $M_2\Vdash\alpha_2$, then  each $\alpha_i$ is of non-empty domain and $(M_1 M_2)\Vdash @_\phi(\alpha_1,\alpha_2)$ -- in other words the so-called blueprint of $M$ is indeed a blueprint, provided so are the blueprints of $M_1$, $M_2$. When $M = \la x.M_1$ the blueprint of $M$ is of the form $*(\alpha)$ -- the relation between $\alpha$ and the blueprint of $M_1$ in that case will be clarified by Lemma \ref{abstract_extract}.
\begin{lemma}\label{her_stable}
For all $M\in\lamt$ and forall $a\cdot b\in\dom(M)$:
\begin{enumerate}
\item If $\free(\compdown{M}{a\cdot b})\subseteq\free(M)$ then $\free(\compdown{M}{a\cdot b})\subseteq \free(\compdown{M}{a})$.
\item If $\compdown{M}{a}\Vdash\alpha$ and  $\compdown{M}{a\cdot b}\Vdash\beta$, then $\compdown{\alpha}{b} = \beta$.
\end{enumerate}
\end{lemma}
\begin{proof} The first proposition is a consequence of our bound variable convention (see Section \ref{la_calc}): if $\free(M) = X$, $\free(\compdown{M}{a}) = X'\cup Y$ where $X'\subseteq X$ and $X$, $Y$ are disjoint, then every element of $\free(\compdown{M}{a\cdot b})$ in $X$ is also an element of $X'$. Thus if 
$a\cdot b$ is in the stable part of $M$, then $b$ is also in the stable part of $\compdown{M}{a}$. The second proposition is equivalent to the first.
\end{proof}
\subsection{Extraction of the formulas of a blueprint}\label{subsect_red_def}
\begin{definition}\label{reduc}
The judgment ``$\beta$ is the blueprint obtained by extracting the formula $\phi$ at the address $a$ in the blueprint $\alpha$'', written $\alpha\rhd^a_\phi\beta$, is inductively defined~by:
\begin{enumerate}
\item $\phi\rhd^\varepsilon_\phi\ebp$,
\item if $\alpha\rhd^a_\phi\beta$, 
then  $@_\psi(\gamma,\alpha)\rhd^{(2)\cdot a}_\phi *(\gamma,\beta)$ 
\item if $\alpha\rhd^a_\phi\beta$, 
then $*_{(b,c_1,\dots,c_k)}(\alpha,\gamma_1,\dots,\gamma_k)\rhd^{b\cdot a}_\phi *_{(b,c_1,\dots,c_k)}(\beta,\gamma_1,\dots,\gamma_{n})$.
\end{enumerate} 
\end{definition}
\begin{figure}
\begin{center}
\epsfig{scale=1.2,file=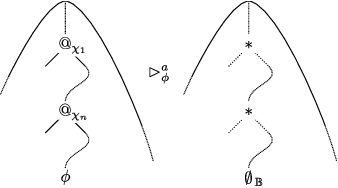}
\end{center}
\caption{\label{blueprint_red_def} Principle of blueprint reduction.}
\end{figure}
In (2) we assume of course that $\alpha$ and $\gamma$ are non-empty. In (3) we assume $b\neq\varepsilon$ in order to avoid circularity.

For instance (Figure~\ref{blueprint_multiple_red_bp_alone}):
\begin{figure}
\begin{center}
\epsfig{scale=1.2,file=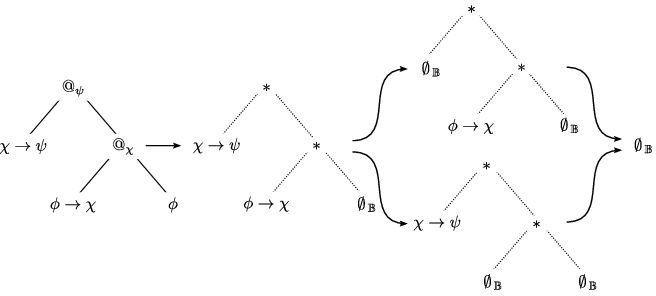}
\end{center}
\caption{\label{blueprint_multiple_red_bp_alone}Full reductions of $@_\psi(\chi\to\psi,@_\chi(\phi\to\chi,\phi))$ to $\ebp$.}
\end{figure}
\begin{itemize}
\item $\begin{array}[t]{lll}
@_\psi(\chi\to\psi,@_\chi(\phi\to\chi,\phi)) & \rhd^{(2,2)}_\phi & *(\chi\to\psi,*(\phi\to\chi,\ebp))\vspace{1ex}\\
&\rhd^{(2,1)}_{\phi\to\chi}& *(\chi\to\psi,*(\ebp,\ebp))\vspace{1ex}\\
&\rhd^{(1)}_{\chi\to\psi}& *(\ebp,*(\ebp,\ebp))= \ebp\vspace{1ex}
\end{array}$\vspace{1ex}
\item $\begin{array}[t]{lll}
@_\psi(\chi\to\psi,@_\chi(\phi\to\chi,\phi))&\rhd^{(2,2)}_\phi& *(\chi\to\psi,*(\phi\to\chi,\ebp))\vspace{1ex}\\
&\rhd^{(1)}_{\chi\to\psi} & *(\ebp,*(\phi\to\chi,\ebp))\vspace{1ex}\\
&\rhd^{(2,1)}_{\phi\to\chi}& *(\ebp,*(\ebp,\ebp))= \ebp
\end{array}$
\end{itemize}
When $\alpha\rhd^a_\phi\beta$, the blueprint $\beta$ can be seen as $\alpha$ in which the formula $\phi$ at $a$ is erased together with all $@$'s in the path to $a$. At each $@$ this path must follow the right branch of $@$. The constraints on the construction of blueprints imply the existence of at least one such path in every non-empty blueprint, even if it is not the blueprint of a term.
\subsection{Sets of extractible sequences}\label{subsect_extract}
\begin{definition}\label{def_linear} For each formula $\phi$, let $\rhd_\phi$ be the relation defined by: $\alpha\rhd_\phi\beta$ if and only if there exists $a$ such that $\alpha\rhd^a_\phi\beta$. We write $\rhd^+_\phi$ for the transitive closure of $\rhd_\phi$. For each  $\alpha$, we write $\linear(\alpha)$ for the set of all sequences $(\phi_1,\dots,\phi_n)$ such that $\alpha\rhd^+_{\phi_n}\dots\rhd^+_{\phi_1}\ebp$. 
\end{definition}
The set $\linear(\alpha)$ is what we called ``set of extractible sequences of $\alpha$'' in the introduction of Section \ref{sect_blue}. Note that $\linear(\ebp) = \{\varepsilon\}$. If $\alpha\neq\ebp$, then all elements of $\linear(\alpha)$ are non-empty sequences. Note also that each $\rhd$-reduction strictly decreases the cardinality of the domain of a blueprint, therefore  $\linear(\alpha)$ is a finite set for all $\alpha$. We now introduce the notion of {\em shuffle} which will allow us to characterise $\linear(\alpha)$ depending on the structure of $\alpha$.
\begin{definition}\label{lin_op} A {\em contraction} of a sequence $F$ is either the sequence $F$ or a sequence $G\cdot(f)\cdot H$ where $G\cdot(f)\cdot(f)\cdot H$  is a contraction of $F$. 
\end{definition}
\begin{definition}
For all finite sequences $F_1, \dots, F_n$ we call {\em shuffle} of $(F_1,\dots, F_n)$ every sequence $F^1_1\cdot\dots\cdot F^1_{n}\cdot\dots\cdot F^p_1\cdot\dots \cdot F^p_n$ such that $F^1_i\cdot\dots \cdot F^p_i = F_i$ for each $i$.  For each tuple of sets of finite sequences $({\cal F}_1,\dots, {\cal F}_n)$ we write $\circledast({\cal F}_1,\dots, {\cal F}_n)$ for the closure under contraction of the set of shuffles of elements of ${\cal F}_1\times\dots\times {\cal F}_n$.
\end{definition}
\begin{definition}
Given two non-empty finite sequences $F_1, F_2$, we call {\em right-shuffle} of $(F_1,F_2)$ every sequence  $F^1_1\cdot F^1_2\cdot\dots\cdot F^p_1\cdot F^p_2$ such that $F^1_i\cdot\dots F^p_i = F_i$ for each $i$ and $F^p_2\neq\varepsilon$. For each pair of sets of non-empty finite sequences $({\cal F}_1,{\cal F}_2)$  we write $\circledcirc({\cal F}_1,{\cal F}_2)$ for the closure under contraction of the set of right-shuffles of elements ${\cal F}_1\times{\cal F}_2$.
\end{definition}
\begin{figure}
\begin{center}
\epsfig{scale=1.2,file=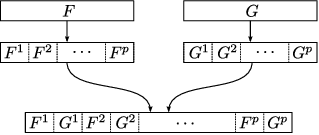}
\end{center}
\caption{\label{sequences_shuffle} Shuffling of two sequences. The chunks of $F$ and $G$ need not to be of the same size -- some of them can be empty. Every contraction of the resulting sequence belongs to $\circledast(F,G)$.  Each contraction belongs also to $\circledcirc(F,G)$ when $F, G$ are non-empty and the last chunk $G^p$ of $G$ is non-empty.}
\end{figure}
 The principle of (right-)shuffling is depicted on Figure \ref{sequences_shuffle}.
The following properties follow from our definitions and will be used without reference:
\begin{enumerate}
\item If $\alpha = \ebp$, then $\linear(\alpha) = \{\varepsilon\}$.
\item If $\alpha =\phi$, then $\linear(\alpha) = \{(\phi)\}$. 
\item If $\alpha = *_{\ov a}(\alpha_1,\dots,\alpha_k)$, then $\linear(\alpha) = \circledast(\linear(\alpha_1),\dots,\linear(\alpha_k))$.
\item If $\alpha = @_\phi(\alpha_1,\alpha_2)$, then $\linear(\alpha) = \circledcirc(\linear(\alpha_1),\linear(\alpha_2))$.
\end{enumerate}
\subsection{Abstraction vs. extraction}\label{subsect_abstr_red}
\begin{lemma}\label{confl_red}
Suppose $\{a_1,\dots,a_p\} = \{b_1,\dots, b_p\}$, and: 
\begin{itemize}
\item $\alpha \rhd^{a_1}_\chi\dots\rhd^{a_p}_\chi\beta$,
\item  $\alpha \rhd^{b_1}_\chi\dots\rhd^{b_p}_\chi\beta'$.
\end{itemize} 
Then $\beta = \beta'$.
\end{lemma}
\begin{proof} By an easy induction on $\alpha$.
\end{proof}
Recall that for every strictly increasing sequence of variables $X = (x_{1},\dots, x_{n})$, we write $\form(X)$ for the sequence of the types of $x_{1},\dots, x_{n}$. We now clarify the link between the blueprint  $\alpha$ of  a term $M$ and the one of $\la x.M$. 

The next lemma shows in particular that if $M, \la x.M\in\lamt$, then $M$ and $\la x.M$  are of blueprints $\alpha$ and $\beta$ if and only if there exist $a_0,\dots,a_p$ such that $\{a_0,\dots,a_p\} = \{a\,|\,\compdown{M}{a}~=~x\}$,  $\alpha\rhd^{a_0}_\chi\dots\rhd^{a_p}_\chi\alpha'$ and $\beta = *(\alpha')$ (Figure \ref{blueprint_abstr}).
\begin{figure}
\begin{center}
\epsfig{scale=1.2, file=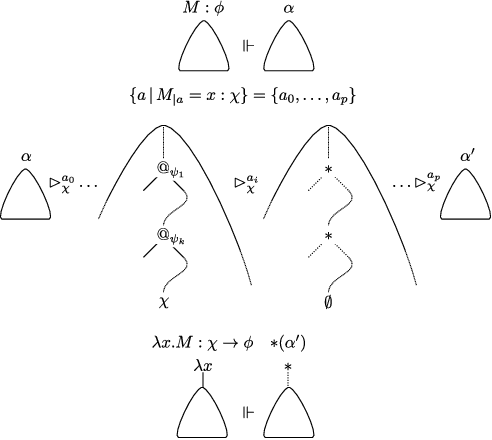}
\end{center}
\caption{\label{blueprint_abstr} How the blueprint of $M$ evolves into the blueprint of $\la x.M$}
\end{figure}
\begin{lemma}\label{abstract_extract} Suppose $M\in\lamt$ is of blueprint $\alpha$,  with $\free(M) = (x_1,\dots, x_n)$ and $\Omega(x_1,\dots, x_n) = (\chi_1,\dots,\chi_n)$. For each $i\in[0,\dots, n]$:
\begin{itemize}
\item let $\alpha_i$ be the restriction of $\alpha$ to $\dom(\alpha) \cap \{a\,|\,\free(\compdown{M}{a})\subseteq \{x_{1},\dots, x_{i}\}\}$.
\item let $\beta_i$ be the blueprint of $\la x_{i + 1}\dots x_n.M$,
\end{itemize}
Then:
\begin{enumerate}
\item For each $i\in[0,\dots,n]$ we have $\dom(\beta_i) = \{1^{n-1}\cdot a\,|\,a\in\dom(\alpha_i)\}$ and $\compdown{\beta_i}{1^{n-i}} = \alpha_i$.
\item For each $i\in\,\,]0,\dots, n]$:
\begin{enumerate}\parskip -1ex
\item there exist $a^i_0,\dots,a^i_{p_i}$ such that $\{a^i_0,\dots,a^i_{p_i}\} = \{a\,|\,\compdown{M}{a} = x_{i}\}$ 

and $\alpha_i\rhd^{a^i_0}_{\chi_i}\dots\rhd^{a^i_{p_i}}_{\chi_i}\alpha_{i-1}$,
\item if $\{b_0,\dots,b_{p_i}\} = \{a\,|\,\compdown{M}{a} = x_{i}\}$ and $\alpha_i\rhd^{b_0}_{\chi_i}\dots\rhd^{b_{p_i}}_{\chi_i}\alpha'$ then $\alpha'=\alpha_{i-1}$.
\end{enumerate}
\item We have $(\chi_1,\dots, \chi_n)\in\linear(\alpha)$.
\end{enumerate} 
\end{lemma}
\begin{proof} Property (1) follows immediately from the definition of a blueprint. Since $\alpha_n~=~\alpha$ and $\alpha_0 = \ebp$, Property (3) follows from Property (2.a). Property (2.b) follows from Property (2.a) and Lemma  \ref{confl_red}.
As to prove (2.a) we introduce the following notations. 

For each $N\in\lamt$, we let $\rho_N$ be the least partial function satisfying the following conditions:  for every blueprint $\gamma$, we have $\rho_N(\varepsilon,\gamma) = \gamma$; for every finite sequence of variables $Y$ and for every blueprint $\gamma$, if  $\rho_N(Y,\gamma) = \delta$, $\{b\,|\,\compdown{N}{b} = y\} = \{b_0,\dots, b_m\}$ and $\delta\rhd^{b_0}_{\chi}\dots \rhd^{b_m}_{\chi}\delta'$, then  $\rho_M((y)\cdot Y,\gamma) = \delta'$.
By Lemma \ref{confl_red}, if $\{b\,|\,\compdown{N}{b} = y\} = \{b_0,\dots, b_m\} = \{c_0,\dots,c_m\}$, $\delta\rhd^{b_0}_{\chi}\dots \rhd^{b_m}_{\chi}\delta'$ and $\delta\rhd^{c_0}_{\chi}\dots \rhd^{c_m}_{\chi}\delta''$, then $\delta' = \delta''$, thus $\rho_N$ is indeed a function. 
For each finite sequence of variables $Y'$ and for each blueprint $\gamma$, we let $\mu_N(Y',\gamma)$ be the restriction of $\gamma$ to $\dom(\gamma) \cap \{b\,|\,\free(\compdown{N}{b})\subseteq Y'\}$. 

We shall prove by induction on $M$ that for all pairs $(X,X')$ such that $\free(M) = X\cdot X'$, we have  $\mu_M(X,\alpha) = \rho_N(X',\alpha)$ -- in particular 
for all $i > 0$ we have 
$$\begin{array}{ll}
\alpha_{i-1} & = \mu_M((x_1,\dots, x_{i-1}),\alpha)\\
&= \rho_M((x_{i},\dots,x_{n}),\alpha)\\ 
&= \rho_M((x_{i}),\rho_M((x_{i+1}\dots,x_{n}),\alpha))\\
&= \rho_M((x_{i}),\mu_M((x_{1}\dots,x_{i}),\alpha))\\
&= \rho_M((x_{i}),\alpha_i)\end{array}$$
thus (2.a) holds. The case $X' = \varepsilon$ is immediate, hence we may as well assume that $X'$ is a non-empty suffix of $\free(M)$. The case of $M$ equal to a variable follows immediately from our definitions.

Suppose $M = (M_1M_2)$, $M_1\Vdash\gamma_1$ and $M_2\Vdash\gamma_2$.
There exist $X_1, X_2, X'_1, X'_2$ such that: $X_1\cup X_2 = X$; $X'_1\cup X'_2 = X'$; $\free(M_j) = X_j\cdot X'_j$ for each $j\in\{1,2\}$. We have $\alpha = @_\psi(\gamma_1,\gamma_2)$ where $\psi$ is the type of $M$, and $\mu_M(X,\alpha) = \ast(\mu_{M_1}(X_1,\gamma_1),\mu_{M_2}(X_2,\gamma_2))$. By induction hypothesis $\mu_{M_i}(X_i,\gamma_i) = \rho_{M_i}(X'_i,\gamma_i)$ for each $i$. The sequence $X'$ is non-empty hence the last elements of $X',X'_2$ are equal. Assume $X' = X''\cdot(x)$ and $X'_2 = X''_2\cdot(x)$. If $x$ is not the last element of $X'_1$ then:  
$$\begin{array}{ll}
\rho_M(X',\alpha) & = \rho_M(X''\cdot(x),@_\psi(\gamma_1,\gamma_2))\\
 & =  \rho_M(X'_1\cup X''_2,*(\gamma_1,\rho_{M_2}((x),\gamma_2)))\\ 
 & = *(\rho_{M_1}(X'_1,\gamma_1),\rho_{M_2}(X''_2,\rho_{M_2}((x),\gamma_2)))\\
& = *(\rho_{M_1}(X'_1,\gamma_1),\rho_{M_2}(X''_2\cdot(x),\gamma_2))\\
& = *(\rho_{M_1}(X'_1,\gamma_1),\rho_{M_2}(X'_2,\gamma_2))
\end{array}$$
Otherwise,  $X'_1 = X''_1\cdot(x)$ and we have:
$$\begin{array}{ll}
\rho_M(X',\alpha) & = \rho_M(X',@_\psi(\gamma_1,\gamma_2))\\
 & =  \rho_M(X''_1\cup X''_2,*(\rho_{M_1}((x),\gamma_1),\rho_{M_2}((x),\gamma_2)))\\ 
 & = *(\rho_{M_1}(X''_1,\rho_{M_1}((x),\gamma_1)),\rho_{M_2}(X''_2,\rho_{M_2}((x),\gamma_2)))\\
& = *(\rho_{M_1}(X''_1\cdot(x),\gamma_1),\rho_{M_2}(X''_2\cdot(x),\gamma_2))\\
& = *(\rho_{M_1}(X'_1,\gamma_1),\rho_{M_2}(X'_2,\gamma_2))
\end{array}$$
In either case 
$$\begin{array}{ll}
\rho_M(X',\alpha) & = \ast(\rho_{M_1}(X'_1,\gamma_1),\rho_{M_2}(X'_2,\gamma_2))\\
& = \ast(\mu_{M_1}(X_1,\gamma_1),\mu_{M_2}(X_2,\gamma_2))\\
& = \mu_M(X,\alpha)\end{array}$$
Suppose $M = \la x.M_1$, $M_1\Vdash\gamma_1$. 
By induction hypothesis $\mu_{M_1}(X,\gamma_1) = \rho_{M_1}(X'\cdot(x),\gamma_1) = \rho_{M_1}(X',\rho_{M_1}((x),\gamma_1)) = \rho_{M_1}(X',\mu(X\cdot X',\gamma_1)) = \rho_{M_1}(X',\compdown{\alpha}{(1)})$. Moreover $\mu_{M_1}(X,\gamma_1) = \mu_{M_1}(X,\mu_{M_1}(X\cdot X',\gamma_1))  = \mu_{M_1}(X,\compdown{\alpha}{(1)})$. Hence $\mu_{M_1}(X,\compdown{\alpha}{(1)}) = \rho_{M_1}(X',\compdown{\alpha}{(1)})$, therefore  $\mu_{M_1}(X,\alpha) = \rho_{M_1}(X',\alpha)$.
\end{proof}
Thus the full sequence of the types of the free variables of $M$ can be extracted from its blueprint. The next lemma shows that conversely for each sequence $\ov\chi$ in $\linear(\alpha)$, there exists a term $N$ with the same domain, blueprint and of the same type as $M$, and such that the sequence of types of the free variables of $N$ is equal to $\ov\chi$, see Figure \ref{blueprint_renaming}.
\begin{figure}
\begin{center}
\epsfig{scale=1.2, file=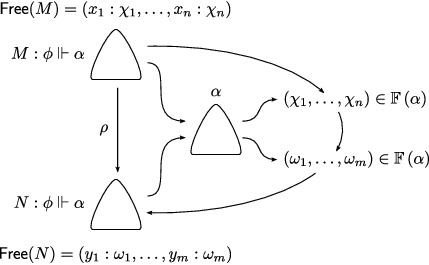}
\end{center}
\caption{\label{blueprint_renaming} A non-uniform renaming of the variables of $M$, based on an alternate extraction of the formulas of its blueprint.}
\end{figure}
\begin{lemma}\label{switch_var} Let $M\in\lamt$ be a term of blueprint $\alpha$. Suppose 
$$\alpha\rhd^{b^m_0}_{\omega_m}\dots\rhd^{b^m_{p_m}}_{\omega_m}\ \ \dots\ \ \rhd^{b^1_0}_{\omega_1}\dots\rhd^{b^1_{p_1}}_{\omega_1}\ebp $$
Then for every strictly increasing sequence of variables $Y = (y_1,\dots, y_m)$ such that  $\form(Y) = (\omega_1,\dots,\omega_m)$, there exists $N$ with the same domain, blueprint and of the same type as $M$ such that $\free(N) = Y$ and $\{b\,|\,\compdown{N}{b} = y_i\} = \{b^i_1,\dots, b^i_{p_i}\}$ for each $i$.  
\end{lemma}
\begin{proof} By induction on $M$. The proposition is clear if $M$ is a variable. The case of $M = (M_1 M_2)$ follows easily from the induction hypothesis. 
Suppose $M = \la x.M_1:\phi\to\psi$ with $M_1\Vdash\gamma$. Let $Y' = (y_1,\dots, y_m,x)$.  By Lemma \ref{abstract_extract}.(2.a) there exist $a_1,\dots, a_p$ such that $\{a_1,\dots, a_p\} = \{a\,|\,\compdown{M}{a} = x\}$ and $\gamma\rhd^{a_0}_{\phi}\dots\rhd^{a_p}_{\phi}\gamma' = \compdown{\alpha}{1}$. Now $$\alpha\rhd^{b^m_0}_{\omega_m}\dots\rhd^{b^m_{p_m}}_{\omega_m}\ \ \dots\ \ \rhd^{b^1_0}_{\omega_1}\dots\rhd^{b^1_{p1}}_{\omega_1}\ebp$$ hence each $b^i_j$ is of the form $(1)\cdot c^i_j$. Furthermore $$\gamma\rhd^{a_0}_{\phi}\dots\rhd^{a_p}_{\phi}\rhd^{c^m_0}_{\omega_m}\dots\rhd^{c^m_{p_m}}_{\omega_m}\ \ \dots\ \ \rhd^{c^1_0}_{\omega_1}\dots\rhd^{c^1_{p_1}}_{\omega_1}\ebp$$
By induction hypothesis there exists $N_1$ with the same domain, blueprint and of the same type as $M_1$ such that $\free(N_1) = Y'$, $\{a\,|\,\compdown{N_1}{a} = x\} = \{a_0,\dots, a_p\}$ and $\{c\,|\,\compdown{N_1}{c}~=~y_i\} = \{c^i_0,\dots, c^i_{p_i}\}$ for each $i$. By Lemma \ref{abstract_extract}.(2.b) we have $\lambda x.N_1\Vdash \alpha$, hence we may take $N = \la x.N_1$.
\end{proof}
\section{Vertical compressions and compact terms}\label{sect_compact}
The aim of this section is to provide a partial characterisation of minimal inhabitants. Section \ref{subsect_bound_local} is just a simple remark on the relative depths of their blueprints, and an easy consequence of the subformula property (Lemma \ref{normal_propto}): if $M$ is a minimal $\lamt$-inhabitant of $\phi$, then for all addresses $a$ in $M$ the blueprint of $\compdown{M}{a}$ is of relative depth at most $k\times p$, where:
\begin{itemize}
\item $k$ is the number of $\la$ in the path from the root to $M$ to $a$,
\item $p$ is the number of subformulas of $\phi$.  
\end{itemize}
We call {\em locally compact} every $\lamt$-inhabitant satisfying this condition.
In Section \ref{subsect_comp} we introduce the notion of {\em vertical compression} of a blueprint. A (strict) vertical compression of $\beta$ is obtained by taking any address $b$ in $\beta$, then by grafting $\compdown{\beta}{b}$ at any address $a < b$ such that $\beta(a) = \beta(b)$. The vertical compressions of $\beta$ are all blueprints obtained by applying this transformation to $\beta$ zero of more times. The key property of those compressions is the following (see Figure \ref{compression_step}):
\begin{figure}
\begin{center}
\epsfig{scale=1.2, file=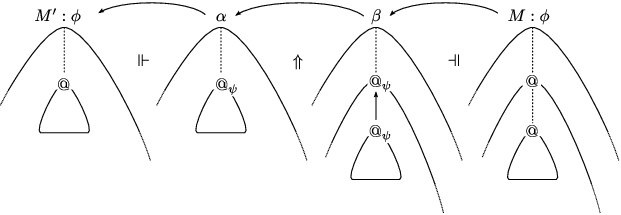}
\caption{\label{compression_step} How the compression of terms is able to follow the compression of blueprints.}
\end{center}
\end{figure}
\begin{itemize}
\item If $M$ is of blueprint $\beta$ and $\alpha$ is a vertical compression of $\beta$, the compression of $\beta$ into $\alpha$ can be mimicked by a compression of $M$ into an HRM-term, in the following sense. Assuming $\alpha = \beta[a\leftarrow\compdown{\beta}{b}]$ (the base case), the term $Q = M[a\leftarrow \compdown{M}{b}]$ is {\em not} in general an HRM-term. However, {\em there exists} an HRM-term $M'$ with the same domain as $Q$ and of the same type as $M$.
Moreover $M'$ and $M$ are applications of the same type or abstractions of the same type.     
\end{itemize}
Let us again consider a $\lamt$-inhabitant $M$ and two addresses $a,b$ such that $a < b$, $\compdown{M}{a}$ and $\compdown{M}{b}$ are applications of the same type or abstractions of the same type. Suppose:
\begin{itemize}
\item there exists a vertical compression $\alpha'$ of the blueprint of $\compdown{M}{b}$ such that the sequence $\Omega(\free(\compdown{M}{a}))$ can be extracted from $\alpha'$.
\end{itemize} 
This situation is a generalisation of the last example in the introduction of Section  \ref{sect_blue} (in which $\alpha'$ was equal to the blueprint of $\compdown{M}{b}$, thereby a trivial compression of this blueprint). The term $M$ is not minimal. Indeed, the key property above implies the existence of a term $N$ of blueprint $\alpha'$ whose size is not greater than the size of  $\compdown{M}{b}$, and such that $N, \compdown{M}{b},\compdown{M}{a}$ are applications of the same type or abstractions of the same type. By Lemma \ref{switch_var}, there exists a term $P$ of the same type and with the same domain as $N$ such that $\free(P) = \free(\compdown{M}{a})$. The graft of $P$ at $a$ yields an inhabitant of strictly smaller size. 

We will call {\em compact} all inhabitants in which the preceding situation does not occur. All inhabitants of minimal size are of course compact. As we shall see in Section  \ref{sect_finite}, we will not need a sharper characterisation of minimal inhabitants. For every formula $\phi$, the set of compact inhabitants of $\phi$ is actually a {\em finite} set, and our decision method will consist in the exhaustive computation of their domains.
\subsection{Depths of the blueprints of minimal inhabitants}\label{subsect_bound_local}
\begin{definition} Two terms $M,M'\in\lamt$ are {\em of the same kind} if and only if they are both variables, or both applications, or both abstractions, and if they are of the same type.
\end{definition}
\begin{definition}\label{def_subform} For all formulas $\phi$, we write $\sub(\phi)$ for the set of all subformulas of $\phi$.
\end{definition}
\begin{definition}\label{def_lambda} Let $M\in\lamt$. Let $a$ be any address in $M$. Let $(a_1,\dots, a_m)$ be the strictly increasing sequence of all prefixes of $a$. Let $(\la x_1,\dots, \la x_k)$ be the subsequence of $(M(a_1),\dots, M(a_{m}))$ consisting of all labels of the form $\la x$. We write $\Lambda(M,a)$ for $(x_1,\dots, x_k)$.
\end{definition}
\begin{definition}\label{def_loc_comp}
Let $M$ be a $\lamt$-inhabitant of $\phi$. We say that $M$ is {\em locally compact} if for all addresses $a$ in $M$, the blueprint of $\compdown{M}{a}$ is of relative depth at most $|\Lambda(M,a)|\times|\sub(\phi)|$.  
\end{definition}
\begin{lemma}\label{local_depth} Let $M$ be a $\lamt$-inhabitant of $\phi$. If $M$ is not locally compact, then there exist two addresses $b$, $b'$ such that  $b < b'$,  $\compdown{M}{b}$ and $\compdown{M}{b'}$ are of the same kind and $\free(\compdown{M}{b}) = \free(\compdown{M}{b'})$. Moreover, $M$ is not a $\lamt$-inhabitant of $\phi$ of minimal size.
\end{lemma}
\begin{proof}
\begin{figure}
\begin{center}
\epsfig{scale=1.2, file=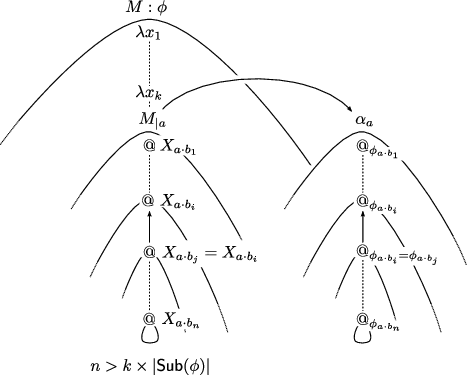}
\end{center}
\caption{\label{compact_case_1} Proof of Lemma \ref{local_depth}.}
\end{figure} For each address $a$ in $\dom(M)$, let $\alpha_a$ be the blueprint of $\compdown{M}{a}$ and let  $X_{a} = \free(\compdown{M}{a})$. Assume the existence of an $\alpha_a$ of relative depth $n >|\Lambda(M,a)|\times|\sub(\phi)|$. There exist $b_1,\dots, b_{n+1}\in\dom(\alpha_a)$ such that  $b_1 < \dots < b_n < b_{n+1}$. By Lemma \ref{her_stable}.(1) we have $X_{a\cdot b_n}\subseteq\dots\subseteq X_{a\cdot b_1}\subseteq \Lambda(M,a)$. By Lemma \ref{normal_propto}, each $\phi_{a\cdot b_i}$ is a subformula of $\phi$. Hence there exist $i, j$ such that $i < j$ and $(X_{a\cdot b_i},\phi_{a\cdot b_i}) = (X_{a\cdot b_j},\phi_{a\cdot b_j})$, that is, $\compdown{M}{a\cdot b_i}$ and  $\compdown{M}{a\cdot b_j}$ are applications of the same type and with the same free variables (Figure \ref{compact_case_1}). Now, let $M' = M[a\cdot b_i\leftarrow \compdown{M}{a\cdot b_j}]$. The term $M'$ is a  $\lamt$-inhabitant of $\phi$ of strictly smaller size.
\end{proof}
\subsection{Vertical compression of a blueprint}\label{subsect_comp}
\begin{definition}\label{def_pumping} We let $\vpump$ be the least reflexive and transitive binary relation on blueprints satisfying the following: if $a,b\in\dom(\beta)$, $a < b$ and $\beta(a) = \beta(b)$, then  $\beta[a\leftarrow \compdown{\beta}{b}]\vpump\beta$. 
\end{definition}
\begin{lemma}\label{diag_cc} Suppose $M\in\lamt$, $M:\phi$, $M\Vdash\beta$ and $\alpha\vpump\beta$. There exists a term $M'\in\lamt$ of the same kind as $M$, of blueprint $\alpha$ and such that $|\dom(M')|\leq |\dom(M)|$.
\end{lemma}
\begin{proof} It suffices to consider the case of $\alpha = \beta[a\leftarrow \compdown{\beta}{b}]$ with $a, b\in\dom(\beta)$, $a < b$ and $\beta(a) = \beta(b)$.
We prove the existence of $M'$ by induction on the length of $a$. If $a = \varepsilon$ then $M$ is necessarily an application and $\beta(\varepsilon) =\beta(b) = @_\phi$, hence $\compdown{M}{b}$ is an application of type $\phi$, and we can take $M' = \compdown{M}{b}$. Assume $a\neq\varepsilon$. 

(1) Suppose $M = (M_1M_2)$, $M_1\Vdash\beta_1$, $M_2\Vdash\beta_2$, $a = (i)\cdot a_{i}$ and $b = (i)\cdot b_i$. By induction hypothesis there exists $M'_i$ of blueprint $\alpha_i = \beta_i[a_i\leftarrow\compdown{\beta_i}{b_i}] = \beta_i[a_i\leftarrow\compdown{\beta}{b}] $,  of the same kind as $M_{i}$ and such that $\dom(M'_i) \leq \dom(M_i)$. Let $j = 1$ if $i = 2$, otherwise let $j = 2$. Let $(M'_{j},\alpha_{j}) = (M_{j},\beta_{j})$. Let $X = (x_1,\dots, x_n)$ be the strictly increasing sequence of all variables free or bound in $M'_2$. Let $Y = (y_1,\dots, y_n)$ be a strictly increasing sequence of variables such that $\Omega(X) = \Omega(Y)$ and $y_1$ is greater that or equal to the greatest variable of $M'_1$. Let $M''_2$ be the term obtained by replacing each $x_i$ by $y_i$ in $M'_2$. We can take $M' = (M'_1M''_2)$. 

(2) Suppose $M = \la x.M_1$, $M_1\Vdash\beta_1$, $x:\chi$, $a = (1)\cdot a_1$ and $b = (1)\cdot b_1$. As $a,b\in\dom(\beta)$, we have also $a_1,b_1\in\dom(\beta_1)$. By induction hypothesis there exists $M'_{1}$ of the same kind as $M_{1}$, of blueprint $\alpha_1 = \beta_1[a_1\leftarrow\compdown{\beta_1}{b_1}]$ and such that $\dom(M'_1) \leq \dom(M_1)$. By Lemma \ref{abstract_extract}.(2.a) there exist $\gamma_1, c_0,\dots, c_p$ such that $\{c_0,\dots, c_p\} = \{c\,|\,\compdown{M}{c} = x\}$, $\beta_1\rhd^{c_0}_{\chi}\dots \rhd^{c_p}_{\chi}\gamma_1$ and $\beta = *(\gamma_1)$. Since $a,b\in\dom(\alpha)$, $a_1$ and $c_i$ are incomparable addresses for all $i$. Hence $\alpha_1 =\beta_1[a_1\leftarrow\compdown{\beta_1}{b_1}]\rhd^{c_0}_{\chi}\dots \rhd^{c_p}_{\chi}\gamma_1[a_1\leftarrow\compdown{\beta_1}{b_1}] = \compdown{\beta[a\leftarrow\compdown{\beta}{b}]}{(1)} = \compdown{\alpha}{1}$. By Lemma  \ref{switch_var} there exists a term $M''_1$ of the same type and with the same domain as $M'_1$ such that the greatest variable $y$ free in $M''_1$ is of type $\chi$ and $\{c\,|\,\compdown{M''_1}{c} = y\} = \{c_0,\dots, c_p\}$. By Lemma \ref{abstract_extract}.(2.b) we have $\la y.M''_1\Vdash\alpha$, hence we may take $M' = \la y.M''_1$.  
\end{proof}
\begin{definition}\label{def_compact} A term $M\in\lamt$ is {\em compact} when there are no $a, b, \alpha'$ such that $a< b$, $\compdown{M}{a}$ and $\compdown{M}{b}$ are of the same kind,  $\compdown{M}{b}\Vdash\alpha_b$, $\alpha'\vpump\alpha_b$ and $\Omega(\free(\compdown{M}{a}))\in\linear(\alpha')$.   
\end{definition}
\begin{lemma}\label{compact_terms}  Every $\lamt$-inhabitant of minimal size is compact. Every compact $\lamt$-inhabitant of $\phi$ is locally compact.
\end{lemma}
\begin{proof} Let $M$ by an arbitrary $\lamt$-inhabitant of $\phi$. 

(1) Assume $M$ is not compact. Let $a,b$ be such that  $a< b$, $\compdown{M}{a}$ and $\compdown{M}{b}$ are of the same kind,  $\compdown{M}{b}\Vdash\alpha_b$, $\alpha'\vpump\alpha_b$, $\free(\compdown{M}{a}) = X_a$ and $\Omega(X_a)\in\linear(\alpha')$ (see Figure \ref{minimal_term}). By Lemma \ref{diag_cc} there exists a term $N\in\lamt$ of blueprint $\alpha'$, of the same kind as $\compdown{M}{b}$ and such that $|\dom(N)|\leq|\dom(\compdown{M}{b})|$. By Lemma \ref{switch_var} there exists $P\in\lamt$ of blueprint $\alpha'$, of the same kind as $N$, such that $\dom(P) = \dom(N)$ and $\free(P) = X_a$. The term $M[a\leftarrow P]$ is then a $\lamt$-inhabitant of $\phi$ of smaller size.

(2) Suppose $M$ meets the conditions of Lemma  \ref{local_depth}. Let $\alpha_{b'}$ be the blueprint of $\compdown{M}{b'}$. By Lemma \ref{abstract_extract}.(3) we have $\Omega(\free(\compdown{M}{b})) = \Omega(\free(\compdown{M}{b'})\in\linear(\alpha_{b'})$. Since the relation $\vpump$ is reflexive, $M$ is not compact. 
\begin{figure}
\begin{center}
\epsfig{scale=1.2, file=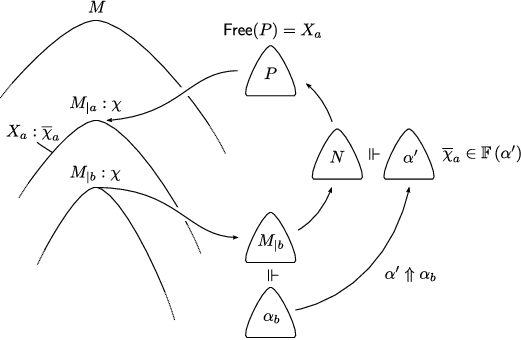}
\end{center}
\caption{\label{minimal_term} Proof of Lemma \ref{compact_terms}, part (1).}
\end{figure}
\end{proof}
\section{Shadows}\label{sect_shadows}
So far we have isolated two properties shared by all minimal inhabitants (Lemma \ref{compact_terms}). We shall now exploit these properties so as to design a decision method for the inhabitation problem. 

In Section \ref{sect_comp} and \ref{sect_shadows_def} we show how to associate, with each locally compact inhabitant $M$ of a formula $\phi$, a tree with the same domain as $M$ which we call the {\em shadow} of $M$.
At each address $a$ this tree is labelled with a triple of the form $(\ov\chi_a,\gamma_a,\phi_a)$ where $\phi_a$  is the type of $\compdown{M}{a}$, the sequence $\ov\chi_a$ is $\Omega(\free(\compdown{M}{a}))$, and $\gamma_a$ is a ``transversal compression'' of the blueprint $\alpha_a$ of $\compdown{M}{a}$ (Definitions \ref{def_equiv} and \ref{mcomp}). Recall that $\ov\chi_a\in\linear(\alpha_a)$ (by Lemma \ref{abstract_extract}.(3)). The blueprint $\gamma_a$ can be seen as a synthesized version of $\alpha_a$ of the same relative depth but of smaller ``width'', and  such that $\ov\chi_a\in\linear(\gamma_a)\subseteq\linear(\alpha_a)$.

Each tree prefix of the shadow of $M$ belongs to a finite set effectively computable from $\phi$ and the domain of this prefix. In particular, one can compute all possible values for its labels, regardless of the full knowledge of $M$ -- or even without the knowledge of the existence of $M$. The key property satisfied by this shadow at every address $a$ is: 
\begin{itemize}
\item for each $\gamma'\vpump\gamma_a$, there exists $\alpha'\vpump\alpha_a$ such that $\linear(\gamma')\subseteq\linear(\alpha')$.
\end{itemize}
This property is sufficient to detect the non-compactness of $M$ for a pair of addresses $(a,b)$ only from the knowledge of $\ov\chi_a, \phi_a,\gamma_b, \phi_b$ and the arity of the nodes at $a$ and $b$. 
Indeed, suppose $a < b$, $\phi_a = \phi_b$ and the nodes at $a$, $b$ are of the same arity (1, or 2). Now, assume:
\begin{itemize}
\item there exists $\gamma'\vpump\gamma_b$ such that $\ov\chi_a\in\linear(\gamma')$.
\end{itemize}
 Then $\compdown{M}{a}$ and $\compdown{M}{b}$  are of the same kind and there exists $\alpha'\vpump\alpha_b$ such that $\ov\chi_a = \Omega(\free(\compdown{M}{a}))\in\linear(\gamma')\subseteq\linear(\alpha')$, therefore $M$ is not compact.

  In Section \ref{sect_shadows_def}, what we call a {\em shadow} is merely a tree $a\mapsto(\ov\chi_a,\gamma_a,\phi_a)$ of a certain shape, no matter if this tree is the shadow of a term or not. This shadow is {\em compact} if there is no pair ($a,b$) as above. Of course, the shadow of a compact term is always compact in this sense.

 In Section \ref{sect_finite} we will prove that for every formula $\phi$, the set of shadows of compact inhabitants of $\phi$  is a finite set effectively computable from $\phi$ (hence the same property holds for the set of compact inhabitants of $\phi$), and we will deduce from this key property the decidability of type inhabitation for HRM-terms.
\subsection{Blueprint equivalence and transversal compression}\label{sect_comp}
\begin{definition}\label{def_equiv} We let $\equiv$ be the least binary relation on blueprints such that:
\begin{enumerate}
\item $\ebp\equiv\ebp$,
\item $\phi\equiv \phi$,
\item if $\alpha_1\equiv\beta_1$, $\alpha_2\equiv\beta_2$, then $@_\phi(\alpha_1,\alpha_1)\equiv@_\phi(\beta_1,\beta_2)$,
\item if $|\ov a| = |\ov b| = n$ and $\alpha_i \equiv\beta_i$ for each $i\in[1,\dots, n]$, then 
$*_{\ov a}(\alpha_1,\dots,\alpha_{n})\equiv*_{\ov b}(\beta_1,\dots,\beta_{n})$.
\end{enumerate}
\end{definition}
In (3), we assume $\alpha_1,\alpha_2,\beta_1,\beta_2$ non-empty. In (4), we assume that the elements of each sequence $\ov a$, $\ov b$ are pairwise incomparable addresses. As to avoid circularity we assume also $a\neq\varepsilon$ or $b\neq\varepsilon$, and $\alpha_i,\beta_i\neq\ebp$ for at least one $i$.

To some extent this equivalence allows us to consider blueprints regardless of the exact values of addresses. For instance  $*_{\ov a}(\alpha_1,\dots,\alpha_{n})\equiv*(\alpha_1,\dots,\alpha_n)\equiv *(\alpha_n,\dots,\alpha_1)$, also $*(*(\alpha,\beta),\gamma) \equiv *(\alpha,\beta,\gamma) \equiv*(\alpha,*(\beta,\gamma))$, etc. It is easy to check that $\alpha\equiv\beta$ implies $\linear(\alpha)=\linear(\beta)$ -- this property will be used without reference.
\begin{definition}\label{mcomp} For each $m\in\nat$, we let $\curvearrowleft_m$ be the least binary relation such that:
\begin{enumerate}
\item if $\gamma_1 \equiv\dots\equiv\gamma_m\equiv\gamma_{m+1}\not\equiv\ebp$, then $\ast_{\ov a}(\gamma_1,\dots,\gamma_m)\curvearrowleft_m\ast_{\ov a\cdot (b)}(\gamma_1,\dots,\gamma_m,\gamma_{m+1})$,
\item if $\alpha = *_{\ov a}(\alpha_1,\dots,\alpha_n)$,  $\beta = *_{\ov b}(\beta_1,\dots,\beta_p)$ and $\alpha\curvearrowleft_m\beta$, then:
\begin{enumerate}
\item $@_\phi(\alpha,\gamma)\curvearrowleft_m @_\phi(\beta,\gamma)$,
\item $@_\phi(\gamma,\alpha)\curvearrowleft_m @_\phi(\gamma,\beta)$,
\item $\ast_{\ov a\cdot(c)}(\alpha_1,\dots,\alpha_n,\gamma)\curvearrowleft_m \ast_{\ov b\cdot(c)}(\beta_1,\dots,\beta_p,\gamma)$.
\end{enumerate}
\end{enumerate}
We call {\em $m$-compression of $\beta$} every $\alpha$ such that $\alpha\curvearrowleft_m\beta$. The {\em width} of $\beta$ is defined as the least $m\in\nat$ for which there is no $\alpha$ such that $\alpha\curvearrowleft_m\beta$.
\end{definition}
Again the elements of $\ov a\cdot(b)$, $\ov a\cdot(c)$ and $\ov b\cdot(c)$ must be pairwise incomparable addresses, and $\alpha,\beta,\gamma$ must be non-empty.
Note that for all non-empty $\beta$, we have $\ebp\curvearrowleft_0\beta$, hence the empty blueprint is the only blueprint of null width. If $\beta$ is of width $m > 0$, then for all addresses $a$, for $\compdown{\beta}{a} = \ast_{\ov a}(\gamma_1,\dots,\gamma_k)$ and for each $\gamma_i\neq\ebp$, the sequence $(\gamma_1,\dots,\gamma_k)$ contains no more than $m$ blueprints $\equiv$-equivalent to $\gamma_i$. For instance, if $\phi,\psi,\chi$ are distinct formulas, $\ast(\phi,\phi,\phi,\psi,\psi,\chi)$ is of width 3, $\ast(\omega, @_\omega(*(\phi,\psi),\phi), @_\omega(*(\psi,\phi),\phi))$ is of width 2, etc.
\begin{definition} \label{ordcomp} For each $m\in\nat$,
we write $\sqeq_m$ for the reflexive and transitive closure of the union of $\equiv$ and $\curvearrowleft_m$. We let $\sqeq^{\max}_m$ denote the subset of the relation $\sqeq_m$ of all pairs with a left-hand-side of width at most $m$.
\end{definition}
For instance, if $\phi,\psi,\chi$ are distinct formulas: 
$$\ebp\sqeq^{\max}_0\ast(\psi,\chi,\phi)\sqeq^{\max}_1\ast(\chi,\phi,\phi,\psi,\psi)\sqeq^{\max}_2\ast(\phi,\phi,\phi,\psi,\psi,\chi)$$
 Of course $\alpha\sqeq_m\beta$ implies $\alpha\sqeq_j\beta$ for all $j\in[1,\dots,m]$ and
 clearly, $\alpha\curvearrowleft_m\beta$ implies $|\dom(\alpha)| < |\dom(\beta)|$, therefore $\curvearrowleft_m$ is well-founded.  
\begin{definition}\label{enum_def} For all ${\cal S}\subseteq\sign$, for all $d\in\nat$ and for all $m\in\nat$:
\begin{itemize}
\item we let $\bp({\cal S},d,\infty)$ be the set of ${\cal S}$-blueprints of relative depth at most $d$,
\item we let $\bp({\cal S}, d, m)$ be the set of all blueprints in $\bp({\cal S}, d, \infty)$ of width at most $m$. 
\end{itemize}
\end{definition}
\begin{lemma}\label{enum_bounded} 
 For all finite ${\cal S}\subseteq\sign$,  for all $d\in\nat$ and for all $m\in\nat$:
\begin{enumerate} 
\item The set $\bp({\cal S}, d, m)/_\equiv$ is a finite set.
\item 
A selector $\brp({\cal S}, d, m)$ for $\bp({\cal S}, d, m)/_\equiv$ is effectively computable from $({\cal S}, d, m)$.
\end{enumerate}
\end{lemma}
\begin{proof}
(1) Let $\bp_\varepsilon({\cal S}, d, m)$ be the set of all rooted blueprints in $\bp({\cal S}, d, m)$. Assuming  $\bp_\varepsilon({\cal S}, d, m)/_\equiv$ is a finite set and a selector $\brp_\varepsilon({\cal S}, d, m)$ for $\bp_\varepsilon({\cal S}, d, m)/_\equiv$ is effectively computable from $({\cal S}, d, m)$, we prove that $\bp({\cal S}, d, m)/_\equiv$ and  $\bp_\varepsilon({\cal S}, d+1, m)/_\equiv$ are finite sets and show how to compute a selector for each set. 

Let $(\alpha_1,\dots,\alpha_k)$ be an enumeration of  $\brp_\varepsilon({\cal S}, d, m)$.  Let $\Sigma_d$ be the set of all functions from $\{1,\dots, k\}$ to $\{0,\dots,m\}$. For each $\beta\in\bp({\cal S}, d, m)$ there exist
$\beta_1,\dots,\beta_n\in\bp_\varepsilon({\cal S}, d, m)$ and  $\ov b$ such that  $\beta=*_{\ov b}(\beta_1,\dots,\beta_n)$. We let $\sigma_\beta$ be the function mapping each $i\in\{1,\dots,k\}$ to the number of occurrences of an element $\equiv$-equivalent to $\alpha_i$ in the sequence $(\beta_1,\dots,\beta_n)$. Clearly $\sigma_\beta\in\Sigma_d$ and furthermore for all $\beta'\in\bp({\cal S}, d, m)$ we have $\beta\equiv\beta'$ if and only if $\sigma_{\beta} = \sigma_{\beta'}$, hence $\bp({\cal S}, d, m)$ is a finite set. 

For each $\tau\in\Sigma_d$, let $\rho_\tau = *(\alpha^1_1,\dots,\alpha^{\tau(1)}_1,\dots,\alpha^1_k,\dots,\alpha^{\tau(k)}_k)$ where each $\alpha^j_i$ is equal to $\alpha_i$. We have $\rho_\tau\in\bp({\cal S}, d, m)$ and $\sigma(\rho_\tau) = \tau$, that is, if $\tau,\tau'\in\Sigma_d$ and $\tau\neq\tau'$, then $\rho_\tau\not\equiv\rho_{\tau'}$. Hence we may define $\brp({\cal S}, d, m)$ as $\{\rho_\tau\,|\,\tau\in\Sigma_d\}$.

The finiteness of  $\bp_\varepsilon({\cal S}, d+1, m)/_\equiv$ follows immediately from the finiteness of $\bp({\cal S}, d, m)$ and the fact that if $\beta = @_\phi(\beta_1,\beta_2)$ and $\beta' = @_\psi(\beta'_1,\beta'_2)$ are elements of $\bp_\varepsilon({\cal S}, d+1, m)$, then $\beta_1,\beta_2,\beta'_1,\beta'_2$ are non-empty elements of $\bp({\cal S}, d, m)$ and furthermore $\beta\equiv\beta'$ if and only if $\beta_1 \equiv \beta'_1$ and $\beta_2\equiv\beta'_2$. The same property allows us to define  $\brp_\varepsilon({\cal S}, d + 1, m)$ as the set of all blueprints of the form $@_\phi(\gamma_1,\gamma_2)$ where $@_\phi\in S$ and each $\gamma_i$ is a non-empty element of $\brp({\cal S},d,m)$.

(2) The lemma follows by induction on $d$, using (1) and the facts that: $\bp_\varepsilon({\cal S}, 0, 0)$ is empty (hence $\bp({\cal S}, d, 0) = \{\ebp\}$ for all $d$); if $m\in\nat_+$, then $\bp_\varepsilon({\cal S}, 0, m)$ is the finite set of all formulas of ${\cal S}$.
 \end{proof}
\subsection{Shadow of a term}\label{sect_shadows_def}
\begin{definition}\label{rep_shadows} Let $\phi$ be a formula. Let ${\cal S}_\phi$ be the union of $\sub(\phi)$ (Definition \ref{def_subform}) and the set of all $@_\psi$ such that $\psi\in\sub(\phi)$.
For each integer $k$, for each formula $\phi$, we let $\shadow(\phi,k) = \brp({\cal S}_\phi, k\times |\sub(\phi)|, k)$, 
where $\brp$ is the function introduced in Lemma \ref{enum_bounded}.(2).
\end{definition}
\begin{definition}\label{shadows} A {\em shadow} is a finite tree in which each node is of arity at most 2 and is labelled with a triple of the form $(\ov\chi, \gamma,\psi)$, where $\ov\chi$ is a sequence of formulas, $\gamma$ is a blueprint and $\psi$ is a formula. 

We call {\em $\phi$-shadow} every shadow $\Xi$ satisfying the following conditions. We have $\Xi(\varepsilon) = (\varepsilon,\ebp,\phi)$. For each $a\in\dom(\Xi)$, let $k_a$ be the number of $b < a$ such that the node of $\Xi$ at $b$ is unary, and let $(\ov\chi_a,\gamma_a,\psi_a) = \Xi(a)$. Then:
\begin{itemize}
\item $\ov \chi_a$ is a sequence of subformulas of $\phi$ of length at most $k_a$,
\item $\gamma_a\in\shadow(\phi,k_a)$,
\item $\ov\chi_a\in\linear(\gamma_a)$
\item $\psi_a$ is a subformula of $\phi$. 
\end{itemize}
\end{definition}
\begin{definition}\label{def_shadow_term} Let $M$ be a locally compact $\lamt$-inhabitant of $\phi$. For each $a\in\dom(M)$:
\begin{itemize}
\item let $\ov\chi_a = \Omega(\free(\compdown{M}{a}))$, 
\item let $\alpha_a$ be the blueprint of $\compdown{M}{a}$, 
\item let $\gamma_a\in\shadow(\phi,|\Lambda(M,a)|)$ be such that $\gamma_a\sqeq^{\max}_{|\Lambda(M,a)|}\alpha_a$,
\item let $\phi_a$ be the type of $\compdown{M}{a}$.
\end{itemize} The tree $\Xi$ mapping each $a\in\dom(M)$ to $(\ov\chi_a.\gamma_a,\phi_a)$ will be called {\em the shadow of $M$}.
\end{definition}
Recall that if $M$ is  a locally compact $\lamt$-inhabitant of $\phi$, then for each address $a$ in $M$, the blueprint $\alpha_a$ of $\compdown{M}{a}$ is of relative depth at most $|\Lambda(M,a)|\times|\sub(\phi)|$. 
Every maximal $|\Lambda(M,a)|$-compression  of $\alpha_a$ produces a shadow $\alpha'_a$ with the same relative depth and of width at most  $|\Lambda(M,a)|$, to which some element of  $\shadow(\phi,|\Lambda(M,a)|)$ is equivalent, thus the shadow of $M$ is well-defined. Note that the choice of $\gamma_a$ is possibly not unique (although it is, since ${\mathbb R}$ is a selector and one can actually prove that $\gamma\sqeq^{\max}_m \alpha$ and $\gamma'\sqeq^{\max}_m \alpha$ implies $\gamma\equiv \gamma'$, but this property is irrelevant to our discussion). We assume that  {\em some} $\gamma_a$ is chosen for each address $a$ in $M$.

Obviously the shadow of $M$ satisfies the first, second and fourth conditions in the definition of $\phi$-shadows given above -- in the next section, we prove that it satisfies also the third.
\begin{figure}
\begin{center}
\epsfig{scale=1.2, file=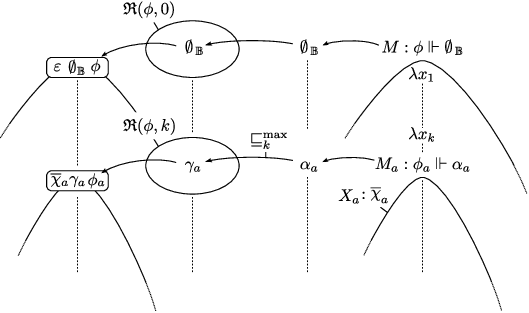}
\caption{\label{shadow_def} A compact inhabitant and its shadow.}
\end{center}
\end{figure}
\subsection{Compact shadows and compact inhabitants}\label{sect_completeness}
\begin{definition}\label{def_compact_shadow} 
A shadow $\Xi$ is {\em compact} if and only if there are no $a, b$ such that: $a < b$, the nodes of $\Xi$ at $a$, $b$ are of the same arity,  $\Xi(a) = (\ov\chi_a,\gamma_a,\psi)$, $\Xi(b)=(\ov\chi_b,\gamma_{b},\psi)$ and there exists $\gamma'\vpump\gamma_b$ such that $\ov\chi_a\in\linear(\gamma')$. 
\end{definition}
Compare this definition with the definition of compactness for term (Definition \ref{def_compact}). With the help of three auxiliary lemmas, we now prove the key lemma of Section \ref{sect_shadows}: if $M$ is a compact inhabitant -- a fortiori locally compact by Lemma \ref{compact_terms} -- then the shadow of $M$ is a compact $\phi$-shadow.
\begin{lemma}\label{commute_left}
 If $\alpha\vpump\beta\sqeq_{1}\beta'$, then there exists $\alpha'$ such that $\alpha\sqeq_{1}\alpha'\vpump\beta'$.
\end{lemma}
\begin{proof} (1) An immediate induction on  $|\dom(\beta')|$ shows that if $\alpha = \beta[a\leftarrow\compdown{\beta}{b}]$ and $\beta\equiv\beta'$, then there exist $a',b'$ such that $a' < b'$ and $\alpha\equiv\alpha' = \beta'[a'\leftarrow\compdown{\beta'}{b'}]$. As a consequence, an immediate induction on the length of the derivation of $\alpha\vpump\beta$ shows that the lemma holds if $\beta\equiv\beta'$. 

(2) Another induction on $|\dom(\beta')|$ shows that if $\alpha\vpump\beta\curvearrowleft_1\beta'$, then there exists $\alpha'$ such that $\alpha\curvearrowleft_1\alpha'\vpump\beta'$. The only non trivial case is $\alpha = \ast_{(a_1)}(\alpha_1)$,  $\beta = *_{(a_1)}(\beta_1)$ with $\alpha_1\vpump\beta_1$ and $\beta' = *_{(a_1,a_2)}(\beta_1,\beta_2)$ with $\beta_1\equiv\beta_2$. Since $\alpha_1\vpump\beta_1\equiv\beta_2$, by (1) there exists $\alpha_2$ such that  $\alpha_1\equiv\alpha_2\vpump\beta_2$. Hence $\alpha = \ast_{(a_1)}(\alpha_1)\curvearrowleft_1\ast_{(a_1,a_2)}(\alpha_1,\alpha_2)\vpump\ast_{(a_1,a_2)}(\beta_1,\beta_2) = \beta'$.

(3) Using (1) and (2), the lemma follows by induction on the length of an arbitrary sequence $(\beta_0,\dots,\beta_n)$ such that $\beta_0 = \beta$, $\beta_n = \beta'$ and $\beta_{i-1}\equiv\beta_{i}$ or  $\beta_{i-1}\curvearrowleft_1\beta_{i}$ for each $i\in[1,\dots,n]$.
\end{proof}
\begin{lemma}\label{preserve_right} If $\alpha\sqeq_1\beta$, then $\linear(\alpha)\subseteq\linear(\beta)$.
\end{lemma}
\begin{proof} By induction on $|\dom(\beta)|$. Since $\gamma\equiv\gamma'$ implies $\linear(\gamma)=\linear(\gamma')$ and $|\dom(\gamma)| = |\dom(\gamma')|$, it suffices to consider the case where $\alpha$ is a $1$-compression of $\beta$.  The case 
$\alpha =  *_{(a_1)}(\alpha_1)$ and 
$\beta = *_{(a_1,a_2)}(\alpha_1,\alpha_2)$ is clear. The remaining cases follow easily from the induction hypothesis.
\end{proof}
\begin{lemma}\label{preserve_left} If $\alpha\sqeq_m\beta$, then the set of all elements of $\linear(\beta)$ of length at most~$m$ is a subset of $\linear(\alpha)$.
\end{lemma}
\begin{proof} By induction on $|\dom(\beta)|$. Again,  we examine only the case $\alpha\curvearrowleft_m\beta$. The proposition is trivially true if $m = 0$. Suppose $m > 0$. The only non-trivial case is
 $\alpha \equiv *_{\ov a}(\gamma_1,\dots,\gamma_m)$ and $\beta \equiv *_{\ov a}(\gamma_1,\dots,\gamma_m,\gamma_{m+1})$ with $\gamma_i \equiv \gamma$ for all $i$. Let $\Phi = \linear(\gamma)$. For each integer $k$, let $\Phi^{(k)} = \circledast( \Phi _1,\dots, \Phi _k)$ where $ \Phi _i = \linear(\gamma)$ for each $i$. Let $\ov\phi = (\phi_1,\dots, \phi_p)\in\linear(\beta)$ be such that $p\leq m$. We have to prove that  $\ov\phi\in\linear(\alpha)$. For each $J\subseteq\{1,\dots, p\}$, let $(j_1,\dots, j_q)$ be the strictly increasing enumeration of all elements of $J$ and let $f(J) =  (\phi_{j_1},\dots, \phi_{j_q})$. We have   $\ov\phi\in\linear(\beta) =  \Phi ^{(m+1)}$, hence there exist $J_1,\dots, J_{m+1}$ such that $J_1\cup\dots\cup J_{m+1} = \{1,\dots, p\}$, and $f(J_i)\in\linear(\gamma)$ for each $i\in \{1,\dots,m+1\}$.
 For each $j\in\{1,\dots, p\}$, let $k_j$ be any element of $\{1,\dots, m+1\}$ such that $j\in J_{k_j}$. Then $J_{k_1}\cup\dots\cup J_{k_p} =\{1,\dots, p\}$, so $\ov\phi\in\circledast(\{f(J_{k_1})\},\dots, \{f(J_{k_p})\})\subseteq \Phi ^{(p)}\subseteq \Phi ^{(m)} = \linear(\alpha)$. 
\end{proof}
\begin{lemma}\label{finiteness}  Let $M$ be a locally compact $\lamt$-inhabitant of $\phi$. The shadow of $M$ is a $\phi$-shadow. If $M$ is compact, then this shadow is also compact.
\end{lemma}
\begin{proof} 
 For each address $a$ in $M$, the sequence $\ov\chi_a = \Omega(\free(\compdown{M}{a}))$ is a subsequence of $\Omega(\Lambda(M,a))$, hence
the first proposition follows from the definition of the shadow of $M$, Lemma \ref{normal_propto}, Lemma \ref{abstract_extract}.(3) and Lemma \ref{preserve_left}. 
Let $\Xi$ be shadow of $M$. Assume $\Xi$ is not compact.
There exist $a, b\in\dom(\Xi) = \dom(M)$ such that $\Xi(a) = (\ov\chi_a,\gamma_a,\psi)$, $\Xi(b) = (\ov\chi_b,\gamma_b,\psi)$, the nodes at $a$,$b$ in $\Xi$ are of the same arity, and there exists $\gamma'\vpump\gamma_b$ such that $\ov\chi_a\in\linear(\gamma')$. We have  $\compdown{M}{a}$, $\compdown{M}{b}$ of the same kind. Let $\alpha_a, \alpha_b$ be the blueprints of  $\compdown{M}{a}$, $\compdown{M}{b}$. Since  $\gamma_{b}\sqeq_{|\Lambda(M,a\cdot b)|}^{\max}\alpha_{b}$, we have $\gamma'\vpump\gamma_b\sqeq_1\alpha_b$. By Lemma \ref{commute_left} there exists $\alpha'$ such that $\gamma'\sqeq_1\alpha'\vpump\alpha_{b}$. By Lemma \ref{preserve_right}, we have $\ov\chi_a\in\linear(\gamma')\subseteq\linear(\alpha')$, hence $M$ is not compact. 
\end{proof}
\section{Finiteness of the set of compact $\phi$-shadows}\label{sect_finite}
Our last aim will be to prove that for each formula $\phi$, the set of all compact $\phi$-shadows is a finite set effectively computable from $\phi$. 

In definition \ref{def_shadow_order}, we introduce a last binary relation $\Subset$ on blueprints. The key lemma of this section (Lemma \ref{extended_well}) shows that whenever ${\cal S}\subset\sign$ is a finite set (in particular when ${\cal S}$ is the set of all subformulas of $\phi$ and all $@$'s tagged with a subformula of $\phi$), the relation $\Subset$ is an almost full relation \cite{Terese2003} on the set of all ${\cal S}$-blueprints: for every infinite sequence $\gamma_1,\gamma_2,\dots$ over $\bp({\cal S})$, there exists $i, j$ such that $i < j$ and $\gamma_i\Subset \gamma_{j}$.
 This result will be proven with the help of Melli\`es' Axiomatic Kruskal Theorem \cite{Mellies1998}. The finiteness of the set of compact $\phi$-shadows follows from this key lemma with the help of  K\"onig's Lemma (Lemma \ref{finite_compact}). The ability to compute these shadows follows directly from their definition.

By Lemma \ref{finiteness}, a consequence of this result is also the finiteness for each $\phi$ of the set of all compact $\lamt$-inhabitants of $\phi$, although our decision method is based on the computation of {\em shadows} of compact terms rather than a direct computation of those terms. It is worth mentioning that the proof of Theorem  \ref{axiom_kruskal} is non-constructive and that it gives no information about the complexity of our proof-search method -- this question might be itself another open problem.
\subsection{Almost full relations and Higman Theorem}\label{sect_akt}
\begin{definition}\label{def_shadow_order} We let $\Subset$ be the relation on blueprints defined by $\alpha\Subset\beta$ if and only if for all $\ov \chi\in\linear(\alpha)$, there exists $\gamma\vpump\beta$ such that $\ov\chi\in\linear(\gamma)$.
\end{definition}
\begin{definition} Let $\cal U$ be an arbitrary set. An {\em almost full relation (AFR) on ${\cal U}$} is a binary relation $\ll$ such that for every infinite sequence $(u_i)_{i\in \nat}$ over ${\cal U}$, there exist $i, j$ such that $i < j$ and $u_i\ll u_j$.
\end{definition}
The main aim of Section \ref{sect_finite} will be to prove the last key lemma from which we will easily infer the decidability of $\lamt$-inhabitation: for each finite ${\cal S}\subseteq\sign$, the relation $\Subset$ is an AFR on~$\bp({\cal S})$. 
\begin{proposition}\label{inter_product_hier}\mbox{}
\begin{enumerate}
\item  If $\ll$ and $\ll'$ are AFRs on ${\cal U}$, then $\ll\cap\ll'$ is an AFR on ${\cal U}$.
\item Suppose $\ll_{\cal U}$ is an AFR on ${\cal U}$ and $\ll_{\cal V}$ is an AFR on ${\cal V}$. Let $\ll_{{\cal U}\times{\cal V}}$ be the relation defined by $(U,V)\ll_{{\cal U}\times{\cal V}}(U',V')$ if and only if $U\ll_{\cal U} U'$ and $V\ll_{\cal V} V'$. Then $\ll_{{\cal U}\times{\cal V}}$ is an AFR on ${\cal U}\times{\cal V}$.
\end{enumerate}
\end{proposition}
\begin{proof} See \cite{Mellies1998}. Both results appear in the proof of Theorem 1, Step 4 (p.523) as a corollary of Lemma 4 (p.520)\end{proof}
\begin{definition} \label{def_Subsetseq}Let ${\cal U}$ be a set, let $\ll$ be a binary relation.
We let $\seq({\cal U})$ denote the set of all finite sequences over ${\cal U}$.
 The relation $\ll_\seq$ {\em induced by $\ll$ on $\seq({\cal U})$} is defined~by $(U_1,\dots, U_n)\ll_\seq(V_1,\dots, V_m)$ if and only if there exists a strictly monotone function $\eta:\{1,\dots, n\}\to\{1,\dots, m\}$ such that $U_i\ll\,V_{\eta(i)}$ for each $i\in\{1,\dots, n\}$.
\end{definition}
\begin{theorem}(Higman) \label{step_B} If $\ll$ is an AFR on ${\cal U}$, then ${\ll}_\seq$ is an AFR on $\seq({\cal U})$.
\end{theorem}
\begin{proof} See \cite{Higman1952, Kruskal1972, Mellies1998}.
\end{proof}
\subsection{From rooted to unrooted blueprints}
Melli\`es' Axiomatic Kruskal Theorem allows one to conclude that a relation is an AFR (a~``well binary relation'' in  \cite{Mellies1998}) as long as it satisfies a set of five properties or ``axioms'' (six in the original version of the theorem -- see the remarks of Melli\`es at the end of its proof explaining why five axioms suffice). The details of those axioms will be given in Section \ref{axiok}. 

 Four of those five axioms are relatively easy to check. The remaining axiom is more problematical. This rather technical section is entirely devoted to the proof of Lemma \ref{root_to_multi}, which will ensure that this last axiom is satisfied. We want to prove the following proposition:
\begin{quote}
{\em  Let ${\cal S}$ be a finite subset of $\sign$. Let ${\cal B}_\varepsilon$ be a subset of $\bpr({\cal S})$. 

Let ${\cal B} = \{*_{\ov a}(\beta_1,\dots,\beta_n)|\,\forall i\in[1,\dots,n],\beta_i\in{\cal B}_\varepsilon\}$. 

If $\Subset$ is an AFR on ${\cal B}_\varepsilon$, then $\Subset$ is an AFR on ${\cal B}$.}
\end{quote}
Recall that  $\bpr({\cal S})$ stands for the set of all rooted ${\cal S}$-blueprints. We want to be able to extend the property that $\Subset$ is an AFR on a given set of rooted blueprints to the set all blueprints that have those rooted blueprints at their minimal addresses. 

Higman Theorem suffices to show that $\Subset_\seq$ (Definition \ref{def_Subsetseq}) is an AFR on the set of finite sequences over ${\cal B}_\varepsilon$. 
However, if one considers an infinite sequence $(\beta_i)_{i\in\nat}$ over ${\cal B}$ and transforms each $\beta_i = *_{\ov a_i}(\beta^i_1,\dots,\beta^i_{n_i})$ where $\beta^i_1,\dots, \beta^i_{n_i}\in{\cal B}_\varepsilon$ into $\sigma(\beta_i) = (\beta^i_1,\dots,\beta^i_{n_i})$, the theorem will only provide two integers $i, j$  and strictly monotone function $\eta$ such that $i < j$ and $\beta^i_k\Subset \beta^j_{\eta(k)}$ for each $k\in\{1,\dots, n_i\}$. This is sufficient to ensure that $\beta_i =  *_{\ov a_i}(\beta^i_1,\dots,\beta^i_{n_i})\Subset \ast_{\ov b}(\beta^j_{\eta(1)},\dots,\beta^j_{\eta(n_i)})$, but not in general $\beta_i\Subset\beta_j$.

To bypass this difficulty we show how for each blueprint $\beta\in\bp({\cal S})$,  one can extract from the set of all vertical compressions of $\beta$ a complete set of ``followers'' of $\beta$ of minimal size (Lemma \ref{sub_linear}). This set $\{\alpha_1,\dots,\alpha_p\}$ has the property that for each $\ov \phi\in\linear(\beta)$, there exists at least one $\alpha_i$ such that $\linear(\alpha_i)$ contains a {\em subsequence} of $\ov\phi$ -- but not necessarily $\ov\phi$ itself. The relative depth of each $\alpha_i$ does not depend on the relative depth on $\beta$, but only on ${\cal S}$: it is at most  $\Sigma_{i=1}^{1+|{\cal S}_@|}\,i$, where ${\cal S}_{@}$ is the set of all binary symbols in ${\cal S}$.
The lemma in proven in four steps. 

First, observe that the set of all $\alpha\vpump\beta$ of relative depth at most $\Sigma_{i=1}^{1+|{\cal S}_@|}\,i$ is a complete set of followers. If we consider the set of all $\gamma$ such that $\gamma\sqeq^{\max}_1\alpha$ for at least one such $\alpha$, we obtain a (possibly infinite) set closed under $\equiv$ and finite up to $\equiv$. We call it the set of {\em ${\cal S}$-residuals} of $\beta$. 

Second, we prove  that the set of ${\cal S}$-residuals of $\beta$ is  a complete set of followers of $\beta$ in the same sense, that is, for each  $\ov \phi\in\linear(\beta)$ there exists an ${\cal S}$-residual $\gamma$ of $\beta$ such that $\linear(\gamma)$ contains a subsequence of $\ov\phi$ (Lemma \ref{preserve_residual}).

Third, we prove that if $\beta = *_{\ov a}(\beta_1,\dots,\beta_n)$, $\beta'= *_{\ov b}(\beta'_1,\dots,\beta'_n,\beta'_{n+1},\dots,\beta'_{n+k})$ are such that $\beta_i\Subset\beta'_i$ for each $i\in[1,\dots,n]$, and if furthermore $\beta, \beta'$ have the same set of ${\cal S}$-residuals, then $\beta\Subset\beta'$ (Lemma \ref{residual}).

The last step is the proof of the lemma itself. The set of ${\cal S}$-residuals is finite up to $\equiv$ (Lemma \ref{enum_bounded}), so there are only a finite number of possible values for the set of residuals of each ${\cal S}$-blueprint. As a consequence, it is always possible to extract from an infinite sequence over ${\cal B}$ an infinite sequence of blueprints with the same set of residuals. The conclusion follows from the third step and Higman Theorem.
\begin{definition} For every ${\cal S}\subseteq\sign$, we let ${\cal S}_{@}$ denote the set of all binary symbols in ${\cal S}$.
\end{definition}
\begin{lemma}\label{sub_linear} Let ${\cal S}$ be a finite subset of $\sign$. For all $\beta\in\bp({\cal S})$, for all $\ov\psi\in\linear(\beta)$, there exists $\alpha$ of relative depth at most $\Sigma_{i=1}^{1+|{\cal S}_@|}\,i$ such that $\alpha\vpump\beta$ and such that $\linear(\alpha)$ contains a subsequence of $\ov\psi$. 
\end{lemma}
\begin{proof} 
Call {\em ${\cal S}$-linearisation} every pair $(\gamma,\ov\chi)$ such that $\gamma\in\bp({\cal S})$ and $\ov\chi\in\linear(\gamma)$. Call {\em starting address for $(\gamma,\ov\chi)$} every address $b$ for which there exist $\phi,\gamma'$ such that $\gamma\rhd^{b}_\phi\gamma'$ and $\ov\chi\in\cc(\linear(\gamma'),(\phi))$.  Call {\em path to $b$ in $\gamma$} the maximal sequence $(b_1,\dots, b_n, b_{n+1})$ over $\dom(\gamma)$ such that $b_1 <\dots < b_n < b_{n+1} = b$. 

Given an arbitrary ${\cal S}$-linearisation $(\beta,\ov\psi)$, we prove simultaneously by induction on $|\dom(\beta)|$ the following properties:
\begin{enumerate}
\item There exists an ${\cal S}$-linearisation $(\gamma,\ov\chi)$ such that:
\begin{enumerate}
\item $\gamma\vpump\beta$ and $\ov\chi$ is a subsequence of $\ov\psi$,
\item $\gamma$ is of relative depth at most  $1 +\Sigma_{i=1}^{|{\cal S}_@|}\,i $.
\end{enumerate}
\item 
There exists an ${\cal S}$-linearisation $(\alpha,\ov\phi)$ such that:
\begin{enumerate}
\item $\alpha\vpump\beta$, $\ov\phi$ is a subsequence of $\ov\psi$, 

and if $\psi\neq\varepsilon$, then the last elements of $\ov\phi, \ov\psi$ are equal,
\item for each starting address $b$ for $(\alpha,\ov\phi)$ and for $(b_1,\dots,b_n,b_{n+1})$ equal to the path to $b$ in $\alpha$, the values $\alpha(b_1),\dots,\alpha(b_n)$ are pairwise distinct,
\item for all $c$ incomparable with each starting address for $(\alpha,\ov\phi)$, 

$(\compdown{\alpha}{c})$ is of relative depth at most  $1 +\Sigma_{i=1}^{|{\cal S}_@|}\,i $. 
\end{enumerate}
\end{enumerate}
Note that the conjunction of (2.b) and (2.c) implies that every address $d$ in $\alpha$ is of relative depth at most $|S_@| + 1 +\Sigma_{i=1}^{|{\cal S}_@|}\,i = \Sigma_{i=1}^{1+|{\cal S}_@|}\,i$. Indeed, suppose $d$ is of maximal relative depth and not a starting address for  $(\alpha,\ov\phi)$. Then $d$ must be incomparable with each starting address for $(\alpha,\ov\phi)$. Let $e$ be the shortest prefix of $d$ in $\dom(\alpha)$ that is incomparable with each starting address for $(\alpha,\ov\phi)$. The address $e$ is of relative depth at most $|{\cal S}_@|$ in $\alpha$ -- otherwise there would exist in $\dom(\alpha)$ an address $f < e$ of relative depth $|{\cal S}_@|$ and a starting adress for $(\alpha,\ov\phi)$ of the form $f\cdot f'$, of relative depth strictly greater than  $|{\cal S}_\alpha|$, a contradiction. Moreover the relative depth of $d$ is the sum of the relative depth of $e$ in $\alpha$ and the relative depth of $\compdown{\alpha}{e}$.

The cases $\beta=\ebp$ is immediate. If $\beta = *_{\ov a}(\beta_1,\dots,\beta_n)$, $i\neq j$ and  $\beta_i,\beta_j\neq\ebp$, then the conclusion follows easily from the induction hypothesis. Suppose $\beta = @_\psi(\beta_1,\beta_2)$. 

(1)  Let $d$ be an address of maximal length in $\beta^{-1}(@_\psi)$.  Let $\delta = @_\psi(\delta_1,\delta_2) = \compdown{\beta}{d}$. By assumption $\varepsilon$ is the only element of $\delta^{-1}(@_\psi)$. As $\ov\psi\in\linear(\beta)$, there exist $\ov\psi_0\in\linear(\delta)$, $\ov\psi_1\in\linear(\delta_1)$, $\ov\psi_2\in\linear(\delta_2)$ such that $\ov\psi_0$ is a subsequence $\ov\psi$ and $\ov\psi_0\in \cc(\{\ov\psi_1\},\{\ov\psi_2\})$. 
By induction hypothesis there exists an $({\cal S} - \{@_\psi\})$-linearisation $(\gamma_1,\ov\chi_1)$ satisfying 
 conditions (1.a), (1.b) w.r.t $(\delta_1,\ov\psi_1)$, and an $({\cal S} - \{@_\psi\})$-linearisation $(\gamma_2,\ov\chi_2)$ satisfying conditions (2.a), (2.b), (2.c) w.r.t $(\delta_2,\ov\psi_2)$. 
Let $\gamma = @_\psi(\gamma_1,\gamma_2)$. We have $\gamma\vpump\delta$ and $\beta(\varepsilon) = \delta(\varepsilon) = \gamma(\varepsilon)$, hence $\gamma\vpump\beta$. The blueprint $\gamma_1$ is of relative depth at most $1 + \Sigma_{i=1}^{|{\cal S}_@| - 1} i\leq\Sigma_{i=1}^{|{\cal S}_@|}\,i$. The blueprint $\gamma_2$ is of relative depth at most $|{\cal S}_@| + \Sigma_{i=1}^{|{\cal S}_@| - 1} =\Sigma_{i=1}^{|{\cal S}_@|}\,i$. Therefore $\gamma$ is of relative depth at most  $1 + \Sigma_{i=1}^{|{\cal S}_@|}\,i$. Now $\ov\chi_2$ is a subsequence of $\ov\psi_2$ with the same last element, so there exists in $\cc(\{\ov\chi_1\},\{\ov\chi_2\})\subseteq\linear(@_\psi(\gamma_1,\gamma_2))$ a subsequence $\ov\chi$ of $\ov\psi_0$. Thus $(\gamma,\ov\chi)$ satisfies (1.a) and (1.b) w.r.t $(\beta,\ov\psi)$.

(2) As $\ov\psi\in\linear(\beta)$, there exist $\ov\psi_1\in\linear(\beta_1), \ov\psi_2\in\linear(\beta_2)$ such that $\ov\psi\in \cc(\{\ov\psi_1\},\{\ov\psi_2\})$. 
By induction hypothesis there exists an ${\cal S}$-linearisation $(\alpha_1,\ov\phi_1)$  satisfying conditions (1.a), (1.b) w.r.t $(\beta_1,\ov\psi_1)$, and
 an ${\cal S}$-linearisation $(\alpha_2,\ov\phi_2)$ satisfying conditions (2.a), (2.b), (2.c) w.r.t $(\beta_2,\ov\psi_2)$. 

 Let $\alpha_0 = @_\psi(\alpha_1,\alpha_2)$. We have $\alpha_0\vpump\beta$. The last elements of $\ov\phi_2$, $\ov\psi_2$ are equal and $\cc(\{\ov\phi_1\},\{\ov\phi_2\})\subseteq\linear(\alpha_0)$. Hence there exists in $\linear(\alpha_0)$ a subsequence $\ov\phi_0$ of $\ov\psi$ with the same last element as $\ov\psi$. Thus $(\alpha_0,\ov\phi_0)$ satisfies (2.a). 

For all $c$ incomparable with each starting address for $(\alpha_0,\ov\phi_0)$, either $c = (1)\cdot c'$ and $c'\in\dom(\alpha_1)$, or $c = (2)\cdot c''$ and $c''\in\dom(\alpha_2)$ is incomparable with each starting address in $\alpha_2$. As a consequence, the choice of $\alpha_1,\alpha_2$ ensures that $(\alpha_0,\ov\phi_0)$ satisfies (2.c).

If $(\alpha_0,\ov\phi_0)$ satisfies (2.b), then we may take $(\alpha, \ov\phi) = (\alpha_0,\ov\phi_0)$. Otherwise some starting address $b$ for $(\alpha_0,\ov\phi_0)$ does not satisfy condition (2.b). Let $(b_1,\dots,b_n,b_{n+1})$ be the path to $b$ in $\alpha$. We have $b_1 = \varepsilon$, and for each $i > 0$, there exists $d_i$ such that $b_i = (2)\cdot d_i$. The sequence $(d_2,\dots,d_{n+1})$ is then a path to $d = d_{n+1}$ in $\alpha_2$, and $d$ is a starting address for $(\alpha_2,\ov\phi_2)$.
 The values $\alpha_2(d_2),\dots,\alpha_2(d_n)$ are pairwise distinct, so there must exist $i > 1$ such that $\alpha(b_i) = @_\psi$. Since $b_i$ is in the path to $b$, there exists in $\linear(\compdown{\alpha_2}{d_i})$ a subsequence $\ov\phi'_0$  of $\ov\phi_0$ with the same last element as $\ov\phi_0$. For $\alpha'_0 = \alpha_0[\varepsilon\leftarrow \compdown{\alpha_2}{d_i}]$, we have $\alpha'_0\vpump\beta$, $\ov\phi'_0\in\linear(\alpha'_0)$ and the last elements of ${\ov\phi}'_0, \ov\phi_0,\ov\psi$ are equal. By induction hypothesis there exists an ${\cal S}$-linearisation  $(\alpha,\ov\phi)$  satisfying (2.a), (2.b), (2.c) w.r.t $(\alpha'_0,\ov\phi'_0)$. The pair $(\alpha,\ov\phi)$ satisfies also  those conditions w.r.t $(\beta,\ov\psi)$.
\end{proof}
\begin{definition}\label{def_residual} Let ${\cal S}$ be a finite subset of $\sign$. For all $\beta\in\bp({\cal S})$, for all $\alpha\vpump\beta$ of relative depth at most $\Sigma_{i=1}^{1+|{\cal S}_@|}\,i$, we call {\em ${\cal S}$-residual} of $\beta$ every $\alpha_0$ such that $\alpha_0\sqeq^{\max}_1\alpha$.
\end{definition}
Note that the set of ${\cal S}$-residuals of $\beta$ is $\{\ebp\}$ if $\beta=\ebp$. Otherwise, it is an infinite set: even if $\beta = \phi$, the set of residuals of $\beta$ is the $\equiv$-equivalence class of $\beta$ and contains all blueprints of the form $*_{a} (\phi)$ (recall that $\equiv$ is a subset of $\sqsubseteq_1$, see Definition \ref{ordcomp}). 
\begin{lemma}\label{preserve_residual}  Let ${\cal S}$ be a finite subset of $\sign$. For all $\beta\in\bp({\cal S})$ and for all $\ov\psi\in\linear(\beta)$, there exists an ${\cal S}$-residual $\alpha_0$ of $\beta$ such that $\linear(\alpha_0)$ contains a subsequence of $\ov\psi$.
\end{lemma}
\begin{proof} (1) Let $\gamma,\delta$ be arbitrary blueprints. Suppose $\gamma\curvearrowleft_1\delta$. We prove by induction on $\delta$ that for all $\ov\phi\in\linear(\delta)$, there exists in $\linear(\gamma)$ a subsequence of $\ov\phi$.
In order to deal with the case $\delta = @_\phi(\delta_1,\delta_2)$, we need to prove a slightly more precise property: for all $\ov\phi\in\linear(\delta)$, there exists in $\linear(\gamma)$ a subsequence $\ov\psi$ of $\ov\phi$ such that the last elements of $\ov\phi$, $\ov\psi$ are equal. The base case is $\delta = *_{(a_1,a_2)}(\gamma_1,\gamma_2)$, $\gamma_1\equiv\gamma_2$ and $\gamma = *_{a_1}(\gamma_1)$, and this case is clear. Other cases follow easily from the induction hypothesis.

(2) We prove the lemma. By Lemma \ref{sub_linear} and by definition of an ${\cal S}$-residual, there exist $\alpha_0, \alpha$ such that $\alpha_0\sqeq_1\alpha\vpump\beta$, $\linear(\alpha)$ contains a subsequence of $\ov\psi$ and $\alpha_0$ is an ${\cal S}$-residual. It follows from (1) that $\linear(\alpha_0)$ contains a subsequence of $\ov\psi$.
\end{proof}
\begin{lemma}\label{residual} Let ${\cal S}$ be a finite subset of $\sign$. Suppose:
\begin{itemize}
\item $\beta =  *_{\ov a}(\beta_1,\dots,\beta_n)\in\bp({\cal S})$, 
\item $\beta' = *_{\ov b}(\beta'_1,\dots,\beta'_n,\beta'_{n+1},\dots,\beta'_{n+k})\in\bp({\cal S})$,
\item  $\beta_i\Subset\beta'_i$ for each $i\in\{1,\dots,n\}$, 
\item  the sets of ${\cal S}$-residuals of $\beta$ and $\beta'$ are equal.
\end{itemize}
Then $\beta\Subset\beta'$.
\end{lemma}
\begin{proof} Let $\ov\psi\in\linear(\beta)$. There exists for each $i\in[1,\dots,n]$ a sequence  $\ov\psi_i\in\linear(\beta_i)$ such that $\ov\psi\in \circledast(\{\ov\psi_1\},\dots, \{\ov\psi_n\})$. By assumption  there exists for each $i\in[1,\dots, n]$ an $\alpha_i\vpump\beta'_i$ such that $\ov\psi_i\in\linear(\alpha_i)$. As a consequence $\ov\psi\in\linear(*(\alpha_1,\dots,\alpha_{n}))$.  

By Lemma \ref{preserve_residual} there exists an ${\cal S}$-residual $\alpha_0$ of $\beta$ such that $\linear(\alpha_0)$ contains a subsequence $\ov\phi$ of $\ov\psi$. By assumption $\alpha_0$ is also an ${\cal S}$-residual of $\beta'$, hence there exist $\alpha'_{1},\dots, \alpha'_{n+k}$, $\ov b$ such that $\alpha_0\sqeq_1 *_{\ov b}(\alpha'_1,\dots,\alpha'_{n+k})\vpump\beta'$.  By Lemma \ref{preserve_right}, we have 
 $\ov\phi\in\linear(\ast_{\ov b}(\alpha'_1,\dots,\alpha'_{n+k}))$. Hence for each $i\in[1,\dots, n+k]$, there exists in $\linear(\alpha'_i)$ a subsequence of $\ov\phi$, which is also a subsequence of $\ov\psi$. 
Now, let  $\alpha = *_{\ov b}(\alpha_1,\dots,\alpha_{n}, \alpha'_{n+1},\dots,\alpha'_{n+k})$.
Then $\alpha\vpump\beta'$, $\ov\psi\in\linear(*(\alpha_1,\dots\alpha_n))$,  and for each $j\in[1,\dots,k]$ there exists in $\linear(\alpha'_{n+j})$ a subsequence of $\ov\psi$. As a consequence $\ov\psi\in\linear(\alpha)$. 
\end{proof}
\begin{lemma}\label{root_to_multi}  Let ${\cal S}$ be a finite subset of $\sign$. Let ${\cal B}_\varepsilon$ be a subset of $\bpr({\cal S})$. Let ${\cal B} = \{*_{\ov a}(\beta_1,\dots,\beta_n)|\,\forall i\in[1,\dots,n],\beta_i\in{\cal B}_\varepsilon\}$. If $\Subset$ is an AFR on ${\cal B}_\varepsilon$, then $\Subset$ is an AFR on ${\cal B}$.
\end{lemma}
\begin{proof} Let ${\cal R} = \bp({\cal S},\Sigma_{i=1}^{1+|{\cal S}_@|}\,i, 1)$ (see Definition \ref{enum_def}). For each $\beta\in{\cal B}$, let $\rho(\beta)$ be the set of ${\cal S}$-residuals of $\beta$. We have  $\rho(\beta)\subseteq {\cal R}$.  Moreover $\rho(\beta)$ is closed under $\equiv$ (as $\equiv$ is a subset of $\sqeq_1$, see Definition \ref{ordcomp}), that is, $\rho(\beta)$ is a union of the elements of a subset of  ${\cal R}/_\equiv$. By Lemma \ref{enum_bounded}.(1) the latter is a finite set, therefore $\{\rho(\beta)\,|\,\beta\in{\cal B\}}$ is a finite set.

For each $\beta = *_{\ov a}(\beta_1,\dots,\beta_n)\in{\cal B}$ where $\ov a$ is increasing w.r.t the lexicographic ordering of addresses and $\beta_1,\dots, \beta_n\in{\cal B}_\varepsilon$, let $\sigma(\beta) = (\beta_1,\dots,\beta_n)$ -- recall that we can take $\ov a = \varepsilon$, $n = 0$ if $\beta = \ebp$, and $\ov a = (\varepsilon)$, $n = 1$ if $\beta$ is a rooted blueprint.  Since  $\{\rho(\beta)\,|\,\beta\in{\cal B\}}$ is a finite set, every infinite sequence over ${\cal B}$ contains an infinite subsequence of blueprints with the same set of ${\cal S}$-residuals.
By assumption $\Subset$ is an AFR on ${\cal B}_\varepsilon$.
By Theorem \ref{step_B}, $\Subset_\seq$ is an AFR on $\{\sigma(\beta)\,|\,\beta\in{\cal B}\}$. 

Thus for every infinite sequence $(\beta_i)_{i\in\nat}$ over ${\cal B}$ there exist $i, j$ such that $i < j$, $\sigma(\beta_i)\Subset_{\seq}\sigma(\beta_{j})$ and $\beta_i$, $\beta_{j}$ have the same set of residuals. For $\sigma(\beta_i) = (\beta^i_1,\dots, \beta^i_n)$ and  $\sigma(\beta_j)= (\beta^j_1,\dots, \beta^j_{n+k})$, there exists a subsequence $(\beta^i_{l_1},\dots, \beta^i_{l_n})$ of $\sigma(\beta_j)$ such that $\beta^i_1\Subset\beta^i_{l_1},\dots,\beta^i_n\Subset\beta^i_{l_n}$. There exist also $l_{n+1},\dots, l_{n+k}$ and two sequences $\ov a$ and $\ov b$ such that  $\beta_i = *_{\ov a}(\beta^i_1,\dots, \beta^i_n)$ and $\beta_j =  *_{\ov b}(\beta^j_{l_1},\dots,\beta^j_{l_n},\beta^j_{l_{n+1}},\dots, \beta^j_{l_{n+k}})$. By Lemma \ref{residual} we have $\beta_i\Subset\beta_j$.
\end{proof}
\subsection{Axiomatic Kruskal Theorem and main key lemma}\label{axiok}
The following definition is borrowed from  \cite{Mellies1998}:

\begin{definition}\label{def_ads} An {\em abstract decomposition system} is an $8$-tuple 
$$({\cal T}, {\cal L}, {\cal V}, \preceq_{\cal T}, \preceq_{\cal L},\preceq_{\cal V}, \stackrel{\cdot}{\longrightarrow}, \vdash)$$ where:
\begin{itemize}
\item ${\cal T}$ is a set of {\em terms} noted $t, u,\dots$ equipped with a binary relation $\preceq_{\cal T}$,
\item ${\cal L}$ is a set of {\em labels} noted $f, g,\dots$ equipped with a binary relation $\preceq_{\cal L}$,
\item ${\cal V}$ is a set of {\em vectors} noted $T, U,\dots$ equipped with a binary relation $\preceq_{\cal V}$,
\item $\stackrel{\cdot}{\longrightarrow}$ is a relation on ${\cal T}\times{\cal L}\times{\cal V}$, {\em e.g.} $t\stackrel{f}{\longrightarrow}T$
\item $\vdash$ is a relation on ${\cal V}\times{\cal T}$, {\em e.g.} $T\vdash t$.
\end{itemize}
For each such system, we let $\rhd_{\cal T}$ be the binary relation on ${\cal T}$ defined by
$$t\rhd_{\cal T} u\Longleftrightarrow \exists(f,T)\in{\cal L}\times{\cal V},\ \ t\stackrel{f}{\longrightarrow}T\vdash u$$
An {\em elementary term} $t$ is a term minimal w.r.t $\rhd_{\cal T}$, that is, a term for which there exists no $u$ such that $t\rhd_{\cal T} u$.
\end{definition}
\begin{theorem}(Melli\`es)\label{axiom_kruskal} Suppose $({\cal T}, {\cal L}, {\cal V}, \preceq_{\cal T}, \preceq_{\cal L},\preceq_{\cal V}, \stackrel{\cdot}{\longrightarrow}, \vdash)$ satisfies the following properties:
\begin{itemize}
\item (Axiom I) There is no infinite chain $t_1\rhd_{\cal T} t_2\rhd_{\cal T}\dots$
\item (Axiom II) The relation $\preceq_{\cal T}$ is an AFR on the set of elementary terms.
\item (Axiom III) For all $t, u, u'$,  

if $t\preceq_{\cal T} u'$ and  $u\rhd_{\cal T} u'$, then $t\preceq_{\cal T} u$.
\item (Axiom IV-bis) For all $t, u, f, g, T, U$, 

if  $t\stackrel{f}{\longrightarrow}T$  and $u\stackrel{g}{\longrightarrow}U$ and $f\preceq_{\cal L} g$ and $T\preceq_{\cal V} U$, then $t\preceq_{\cal T} u$.
\item (Axiom V) For all ${\cal W}\subseteq {\cal V}$, for ${\cal W}_\vdash = \{t\in{\cal T}\,|\,\exists T\in{\cal W}, T\vdash t\}$,  

if $\preceq_{\cal T}$ is an AFR on ${\cal W}_\vdash$, then $\preceq_{\cal V}$ is an AFR on ${\cal W}$.
\end{itemize}
If furthermore $\preceq_{\cal L}$ is an AFR on ${\cal L}$, then $\preceq_{\cal T}$ is an AFR on ${\cal T}$.
\end{theorem}
\begin{proof} See \cite{Mellies1998}. Mellies' result is actually established for an alternate list of axioms (numbered from I to VI). The possibility to drop Axiom VI and to replace Axiom IV with Axiom IV-bis is a remark that follows the proof of the main theorem.
\end{proof}
\begin{lemma}\label{extended_well} For each finite ${\cal S}\subseteq\sign$, the relation $\Subset$ is an AFR on $\bp({\cal S})$.
\end{lemma}
\begin{proof} According to Lemma \ref{root_to_multi} it is sufficient to prove that  $\Subset$ is an AFR on $\bpr({\cal S})$. Let $({\cal T}, {\cal L}, {\cal V}, \preceq_{\cal T}, \preceq_{\cal L},\preceq_{\cal V}, \stackrel{\cdot}{\longrightarrow}, \vdash)$ be the abstract decomposition system defined as follows.
\begin{itemize}
\item The set ${\cal T}$ is $\bpr({\cal S})$; we let $\alpha\preceq_{\cal T}\beta$ if and only if there exists an address $c$ such that $\alpha\Subset(\compdown{\beta}{c})$ and $\alpha(\varepsilon) = (\compdown{\beta}{c})(\varepsilon)$.
\item  The set ${\cal L}$ is the set of all $@$'s in ${\cal S}$, the relation $\preceq_{\cal L}$ is the identity relation on this set. 
\item  The set ${\cal V}$ is $\bp({\cal S})\times\bp({\cal S})$.

The relation  $\preceq_{\cal V}$ is defined by $(\alpha_1,\alpha_2)\preceq_{\cal V}(\beta_1,\beta_2)$ if and only if $\alpha_1\Subset\beta_1$ and $\alpha_2\Subset\beta_2$.
\item The relation $\stackrel{\cdot}{\longrightarrow}$ is defined by $\alpha\stackrel{@_\phi}{\longrightarrow}(\beta_1,\beta_2)$ if and only if $\alpha = @_\phi(\beta_1,\beta_2)$.

\item The relation $\vdash$ is the least relation satisfying the following condition. If $V = (\alpha_1,\alpha_2)$, $i \in\{1,2\}$,  $\beta_1,\dots,\beta_n\in\bpr({\cal S})$  and $\alpha_i = *_{\ov a}(\beta_1,\dots,\beta_n)$, then  $V\vdash\beta_j$ for each $j\in[1,\dots,n]$.
\end{itemize}
Note that the elements of ${\cal V}$ are pairs of blueprints that may be rootless. However if $V\vdash\beta$, then the blueprint $\beta$ is always a rooted blueprint, thus the relation $\vdash$ is indeed a subset of ${\cal V}\times{\cal T}$.

(A) For all ${\cal T}'\subseteq{\cal T}$, the relation $\Subset$ is an AFR on ${\cal T}'$ if and only if $\preceq_{\cal T}$ is an AFR on ${\cal T}'$. Indeed, consider an arbitrary infinite sequence $\ov \alpha$  over ${\cal T}'$. This sequence contains an infinite subsequence $(\alpha)_{i\in\nat}$ such that all $\alpha_i(\varepsilon)$ are equal. Clearly $\alpha_i\Subset\alpha_j$ implies   $\alpha_i\preceq_{\cal T}\alpha_j$. Conversely, if  $\alpha_i\preceq_{\cal T}\alpha_j$, then there exists $c$ such that $\alpha_i\Subset\compdown{\alpha_j}{c}$ and $\alpha_i(\varepsilon) = \alpha_j(\varepsilon) = \alpha_j(c)$. So $\alpha_i\Subset\compdown{\alpha_j}{c}\vpump\alpha_j$, hence $\alpha_i\Subset\alpha_j$. 

(B) We now check that all axioms of Theorem \ref{axiom_kruskal} are satisfied. Axiom I is clear.  The set of elementary terms is the set of all blueprints consisting of single formulas of ${\cal S}$. The relation $\preceq_{\cal T}$ is of course an AFR on the set of elementary terms, that is, axiom II is satisfied. Axiom III is immediate. If $(\alpha_1,\alpha_2)\preceq_{\cal V}(\beta_1,\beta_2)$ then $\alpha_1\Subset\beta_1$ and $\alpha_2\Subset\beta_2$, hence $@_\psi(\alpha_1,\alpha_2)\Subset@_\psi(\beta_1,\beta_2)$, a fortiori $@_\psi(\alpha_1,\alpha_2)\preceq_{\cal T}@_\psi(\beta_1,\beta_2)$, hence  Axiom IV-bis is satisfied. It remains to prove that Axiom V is satisfied. Let ${\cal W}\subseteq {\cal V}$. By definition ${\cal W}_\vdash = \{\beta\in{\cal T}\,|\,\exists (\alpha_1,\alpha_2)\in{\cal W}, (\alpha_1,\alpha_2)\vdash \beta\}$. Assuming $\preceq_{\cal T}$ is an AFR on ${\cal W}_\vdash$, we prove that $\preceq_{\cal V}$ is an AFR on ${\cal W}$. By (A) the relation $\Subset$  is an AFR on ${\cal W}_\vdash\subseteq \bpr({\cal S})$. Let ${\cal B} = \{*_{\ov a}(\beta_1,\dots,\beta_n)|\,\forall i\in[1,\dots,n],\beta_i\in{\cal W}_\vdash\}$.  By Lemma \ref{root_to_multi} the relation $\Subset$ is an AFR on ${\cal B}$. Moreover ${\cal W}\subseteq{\cal B}\times{\cal B}$. By Proposition \ref{inter_product_hier}.(2) the relation $\preceq_{\cal V}$ is an AFR on ${\cal B}\times{\cal B}$, therefore an AFR on ${\cal W}$.
\end{proof}
\begin{lemma}\label{finite_compact} For each formula $\phi$, the set of all compact $\phi$-shadows is a finite set effectively computable from $\phi$.
\end{lemma}
\begin{proof} For each compact $\phi$-shadow $\Xi$ and for each address $a$ such that $a$ is a leaf in $\Xi$, call {\em step-continuation at $a$} of $\Xi$ every compact $\phi$-shadow $\Xi'$ such that $\dom(\Xi') \subsetneq\dom(\Xi)\cup\{a\cdot(1), a\cdot(2)\}$ and $\Xi, \Xi'$ take the same values on $\dom(\Xi)$. Let $\leadsto$ be the relation defined by $\Xi\leadsto \Xi'$ if and only if $\Xi'$ is a step continuation of $\Xi$. By Lemma \ref{enum_bounded} and the fact that the set of subformulas  of $\phi$ is a finite set, for all $\Xi$, the set of all $\Xi'$ such that $\Xi\leadsto\Xi'$, is a finite set effectively computable from $\Xi$.  Let ${\cal C}$ be the closure under $\leadsto$ of $\{(\varepsilon\mapsto(\varepsilon,\ebp,\phi))\}$ The set of all compact $\phi$-shadows is clearly equal to this set, hence it suffices to prove that ${\cal C}$ is a finite set. Assume by way of contradiction that ${\cal C}$ is infinite. By K\"onig's Lemma there exists an infinite sequence $\Xi_0\leadsto \Xi_1\leadsto\dots$ over ${\cal C}$. The union  $\Xi_\infty = \cup_{i \geq 0}\,\Xi_i$ is a tree of infinite domain. By  K\"onig's Lemma again, there exists an infinite chain of addresses $a_1 < a_2 <\dots $ such that all $a_i$ are nodes of  $\Xi_\infty$ with the same arity and labelled with the same subformula of $\phi$. If $i < j$ and $a_i$, $a_j$ are labelled with $(\ov\chi_i,\gamma_i,\psi)$, $(\ov\chi_i,\gamma_j,\psi)$, then we cannot have $\gamma_i\Subset\gamma_j$, otherwise there would exist a $k$ such that  $\Xi_k$ is not compact. A contradiction follows from Lemma \ref{extended_well}.
\end{proof}
\section{From the shadows to the light}
\begin{theorem}\label{decidability} Ticket Entailment is decidable.
\end{theorem}
\begin{proof} The following propositions are equivalent:
\begin{itemize}
\item the formula $\phi$ is provable in the logic $T_\to$,
\item the formula $\phi$ is inhabited by a combinator within the basis ${\sf BB'IW}$,
\item the formula $\phi$ is $\lamt$-inhabited  (Lemma \ref{main_equiv}),
\item there exists a compact $\lamt$-inhabitant of $\phi$  (Lemma \ref{compact_terms})
\item there exists a compact $\phi$-shadow with the same tree domain as a $\lamt$-inhabitant of $\phi$  (Lemmas \ref{compact_terms} and  \ref{finiteness}).
\end{itemize}
 By Lemma \ref{finite_compact}, the set of compact $\phi$-shadows is effectively computable from $\phi$. By the subformula property (Lemma \ref{normal_propto}), for each shadow $\Xi$ in this set, up to the choice of bound variables, there are only a finite number of $\lamt$-inhabitant of $\phi$ with the same domain as $\Xi$. Moreover this set of inhabitants is clearly computable from $\Xi$ and $\phi$. Hence the existence of a $\lamt$-inhabitant of $\phi$ is decidable. 
\end{proof}
\subsection*{Acknowledgments}
This work could not have been achieved without countless helpful comments and invaluable support from Pawe{\l} Urzyczyn, Paul-Andr\'e Melli\`es and Pierre-Louis Curien. I am also deeply indebted to the anonymous referees for their remarkably careful reading.


\begin{thebibliography}{10}
\bibitem[Anderson and Belnap 1975]{AndersonBelnap1975}
Anderson, A.~R., and {Belnap Jr}, N.~D. (1975) \emph{Entailment: The Logic of Relevance and Necessity, Vol.~1}. Princeton University Press.
\bibitem[Anderson 1960]{Anderson1960} 
 Anderson, A.~R. (1960) Entailment shorn of modality. \emph{J. Symb. Log.} \textbf{25}~(4), 388.
\bibitem[Anderson \emph{et al.} 1990]{AndersonBelnapDunn1990}
 Anderson, A.~R., {Belnap Jr}, N.~D., and Dunn, J.~M. (1990) \emph{Entailment: The Logic of Relevance and Necessity, Vol.~2}. Princeton University Press.
\bibitem[Barendregt 1984]{Barendregt1984} Barendregt, H., \emph{The Lambda Calculus: Its Syntax and Semantics},  Studies in Logic and the Foundations of Mathematics, 103 (Revised ed.), North Holland.
\bibitem[Barwise 1977]{Handbook} \emph{Handbook of Mathematical Logic} (1977). Edited by Barwise, J., Studies in Logic and Foundations of Mathematics, North-Holland. 
\bibitem[Bezem, Klop and de Vrijer 2003]{Terese2003}
Bezem, M., Klop, J.,W., de Vrijer, R., (``Terese'') (2003) Term Rewriting Systems.
{\em Cambridge Tracts in Theoretical Computer Science} \textbf{55}, Cambridge University Press.
\bibitem[Bimb\'o 2005]{Bimbo2005}
 Bimb\'o, K. (2005) Types of I-free hereditary right maximal terms. \emph{Journal of  Philosophical Logic} \textbf{34}~(5--6), 607--620.
\bibitem[Broda \emph{et al.} 2004]{BrodaDF2004}
Broda, S., Damas, L., Finger, M., and Silva~e~Silva, P.~S. (2004) The decidability of a fragment of  $BB'IW$-logic. \emph{Theor. Comput. Sci.} \textbf{318}~(3), 373--408.
\bibitem[Bunder 1996]{Bunder1996}
Bunder, M.,W., (1996) Lambda Terms Definable as Combinators.  \emph{Theor. Comput. Sci.} \textbf{169}~(1), 3--21.
\bibitem[Higman 1952]{Higman1952}
Higman, G. (1952) Ordering by divisibility in abstract algebra. \emph{Proc. London Math. Soc.} \textbf{3}~(2), 326--336.
\bibitem[Kripke 1959]{Kripke1959}
Kripke, S (1959) The problem of entailment. \emph{J.  Symb. Log.} \textbf{24}~(4), 324.
\bibitem[Krivine 1993]{Krivine1993} Krivine, J.-L. (1993) \emph{Lambda-calculus, types and models}. Masson.
\bibitem[Kruskal 1972]{Kruskal1972}
Kruskal, J.~B. (1972) The theory of well-quasi-ordering: A frequently discovered  concept. \emph{J. Comb. Theory, Ser. A} \textbf{13}~(3), 297--305.
\bibitem[Melli\`es 1998]{Mellies1998}
Melli{\`e}s, P.-A. (1998) On a duality between Kruskal and Dershowitz theorems. In: Larsen, K.~G, Skyum, S., Winskel, G. (Eds.), ICALP, \emph{Lecture Notes in Computer Science} \textbf{1443}, 518--529, Springer-Verlag.
\bibitem[Trigg \emph{et al.} 1994]{TriggHB1994}
Trigg, P., Hindley, J.~R., and Bunder, M.~W. (1994) Combinatory abstraction using $B$, $B'$ and friends. \emph{Theor. Comput. Sci.} \textbf{135}~(2), 405--422.
\bibitem[Urquhart 1984]{Urquhart1984}
Urquhart, A (1984) The undecidability of entailment and relevant implication. \emph{J.  Symb. Log.} \textbf{49}~(4), 1059--1073.
\end{thebibliography}
\end{document}